\documentclass[10pt]{article}

\usepackage{latexsym}
\usepackage{amsmath,amssymb,pifont,graphicx,epstopdf,float}
\usepackage{comment,epsfig,subfigure,url,latexsym,cite,amsfonts}
\usepackage{algorithm}
\usepackage{algorithmic}
\usepackage{enumerate}

\newtheorem{claim}{Claim}
\newtheorem{lemma}{Lemma}

\newtheorem{theorem}[lemma]{Theorem}

\newcommand{\QED}{\mbox{}\hfill \rule{3pt}{8pt}\vspace{10pt}\par}
\newenvironment{proof}{\noindent \mbox{}{\bf Proof:}}{\QED}

\begin{document}

\title{Multi-skill Collaborative Teams based on Densest Subgraphs}

\begin{titlepage}

\author{
Amita Gajewar \thanks{Yahoo! Inc, Santa Clara, CA, USA. \hbox{E-mail}:~{\tt amitag@yahoo-inc.com}.} 
\and 
Atish {Das Sarma} \thanks{Google Research, Google Inc., Mountain View, CA, USA. \hbox{E-mail}:~{\tt dassarma@google.com}. Part of the work done while at Georgia Institute of Technology, GA.}
}

\date{}

\maketitle \thispagestyle{empty}

\vspace*{.4in}

\maketitle

\begin{abstract}
We consider the problem of identifying a team of skilled individuals for collaboration, in the presence of a social network. Each node in the input social network may be an expert in one or more skills - such as theory, databases or data mining. The edge weights specify the affinity or collaborative compatibility between respective nodes. Given a project that requires a set of specified number of skilled individuals in each area of expertise, the goal is to identify a team that maximizes the collaborative compatibility. For example, the requirement may be to form a team that has at least three databases experts and at least two theory experts.

We explore team formation where the collaborative compatibility objective is measured as the density of the induced subgraph on selected nodes. The problem of maximizing density is NP-hard even when the team requires a certain number of individuals of only one specific skill. We present a 3-approximation algorithm that improves upon a naive extension of the previously known algorithm for densest at least $k$ subgraph problem. We further show how the same approximation can be extended to a special case of multiple skills as well. Our problem generalizes the formulation studied by Lappas et al. [KDD '09]. Further, they measured collaborative compatibility in terms of diameter and the spanning tree costs. Our density based objective also turns out to be more robust in certain aspects. 

Experiments are performed on a crawl of the DBLP graph where individuals can be skilled in at most four areas - theory, databases, data mining, and artificial intelligence. In addition to our main algorithm, we also present heuristic extensions to trade off between the size of the solution and its induced density. These density-based algorithms outperform the diameter-based objective on several metrics for assessing the collaborative compatibility of teams. The solutions suggested are also intuitively meaningful and scale well with the increase in the number of skilled individuals required. 
\end{abstract}



\end{titlepage}


\section{Introduction}
A team formation problem consists of forming a team from a large set of candidates such that the resulting team is best suited to perform the assignment. The main difficulty in providing an automated way to form a team from the solution space is the categorization of the desired attributes quantitatively. In spite of this, the problem has attracted many researchers and various interesting approaches have been suggested over the years, as we mention them in the related work section. In this spirit, we study this problem in the context of social network with a goal to identify the most collaborative team that satisfies the skill-set requirements of the project. Certainly, the naive approach would be to just find the candidates that match the requirements the best. However, considering the social network associated with the candidates add a value to the solution becasue, intuitively, such team is more likely to demonstrate better collaborative compatibility. This is also evident in practice, where many companies tend to promote employee referral program while hiring a candidate.

We model this team formation problem in the social network context by considering the network graph that connects the individuals, wherein each individual is represented by a node in the graph and an association between individuals is represented by an edge in the graph. In a more generic sense, each node can be assigned a set of desired attributes and an edge can be assigned a weight representing the collaborativeness between the individuals it is connecting. Note that, this model could further be extended in multiple dimensions and we believe that the work we present in this paper could be a good starting point with this regard. For example, one possible extension to this graph model would be a hypergraph model wherein we can accomodate many criteria - weight associated with hyperedge could define the colloaborative compatibility between the set of nodes (instead of just two nodes), hyperedge could also be used to denote the set of nodes that represent a certain group, etc. 

In this paper, as a starting point, we define the problem where each node is associated with a set of skills and a weight of the edge reflects the cohesiveness between two connecting nodes(users), and a goal is to form a {\em collaborative} team for a project that requires a specified number of people in each of a set of skills . In this setting, two users can collaborate better as a team if they have a high-weight edge (strong affinity for interaction) between them.  Specifically, consider the following example where a social network of computer scientists is presented. Each user is skilled in a subset of areas between theory, databases and data mining. A company wants to hire people for a predetermined project. The goal of the project requires that the team consists of at least three database researchers, at least two theory researchers, and at least one researcher with expertise in data mining. Presented with the social network where edges reflect collaborative interactions, how should the company go about hiring a team for the project?

A special case of this problem was studied in~\cite{LLT}. They consider team formation when the team requires at most one person each in a set of different skills. Our problem formulation generalizes this by allowing the team to require multiple skilled individuals in any skill. Clearly there are projects where multiple people with specific skills may be desired. It turns out that this generalization makes the problem significantly harder and more interesting. For example, the problem is no longer trivial even when the social network contains users that are either skilled or not skilled in just one specific area. Suppose a project requires eight database researchers, and the social network contains people who are either skilled in databases or not, how does one go about choosing the team? We shall mention the complexity as well as algorithmic results for this special case as well shortly.

A critical question in team formation based on a social network is to determine the collaborative quality of a team. The edges specify the collaborative compatibility of two nodes. However, given a subset of say $k$ nodes in the social network (let us even say these $k$ nodes are connected), how do we know how {\em collaborative} this team is? To tackle this, ~\cite{LLT} suggested two objectives: one based on the diameter of the subgraph induced by these $k$ nodes, and another based on the spanning tree cost of these nodes; and demontrated the potential of these ideas through experimental results. These objectives can certainly be applied to solve the problem we define in this paper. In fact, we provide the extention to their diameter-based algorithm, prove the 2-approximation bound and also complement with experimental results. Similarly, the minimum-spanning tree based approach could also be extended to the problem defined here.
However, the main focus of our paper is a {\em novel} density based objective that we propose for this problem; therefore, the majority of this paper's contributions are related to this density objective. Specifically, we define the collaborative affinity of a team of $k$ nodes to be proportional to the density of the induced subgraph. Using density as a measure of the quality of an induced subgraph of nodes has certain intuitive merits over using diameter or minimum spanning tree costs; we describe these in section~\ref{sec:def}.


We briefly summarize the problem definiton here: given a set of skills $1, 2, \ldots, t$, and requirements $k_1, k_2, \ldots, k_t$, and a social network of nodes connected by (weighted) edges, the goal is to pick a subset of nodes such that at least $k_i$ distinct nodes possess skill $i$, for $1\leq i\leq t$. The same node, however, may contribute to two different skills. The objective value of the solution is the density of the induced subgraph on these nodes.
The goal is to maximize this objective. Notice that the number of returned nodes may be as small as $k_{max} = \max_{i}{k_i}$ or be even larger than $\sum_{i}{k_i}$. We now summarize the contributions of this paper.

\noindent{\bf Our Contributions.}

\begin{itemize}
\item We present a novel problem definition for team formation to maximize collaborative compatibility. The constraint of the problem requires the team to comprise of at least a {\em specified number} of skilled individuals in each of a set of skills. This generalizes previous work that required forming a team with at least {\em one} skilled individual in each of a set of skills.
\item As a measure of collaborative compatibility, we suggest a density based objective. Density is a novel metric for this domain and we show that it has certain desirable properties for measuring compatibility. Our density based team formation problem also generalizes previous graph algorithms work on finding densest subgraphs with size constraints.
\item We address the collaborative team formation problem when the team requires one or more skills. We show that optimizing even the special case of a single skill is  NP-hard under our density-based metric, as well as the previously suggested diameter-based metric. The main theoretical result of the paper is to present a novel 3-approximation algorithm for the density based team formation problem for both single as well as a special case of multiple skills. This improves upon a naive extension of previous work on size constrained densest subgraph problems. We also show how previous work on a 2-approximation for the diameter-based objective can be extended to our generalized problem.
\item We present several heuristic algorithms that build on our 3-approximation for density-based team formation. These algorithms trade-off between the size of the returned solution and the density, while respecting the constraints on the skill requirements. 
\item We perform experiments on all these algorithms on the DBLP graph. Experiments show that density-based algorithms perform well in practice, identifying tightly knit and highly skilled teams and also scale well with the size of the team and skill requirements. 
\item We measure qualitative evidence of the teams reported by both denisty-based and diameter-based algorithms and show that the density-based algorithms compare favorably to the diameter-based algorithms on a number of different metrics. Further analysis of the teams (by inspecting the members of the team) reported show that the density-based approach suggest the teams that are more intuitive and meaningful compared to diameter-based teams.
\end{itemize}

\noindent{\bf Overview.} We mention related work in Section~\ref{sec:rel}. The various problem definitions, notations and some properties are formalized in Section~\ref{sec:def}. Our theoretical contributions, including the main 3-approximation algorithm for our density based objective are described in Section~\ref{sec:theory1}. The theoretical work on a diameter based objective is presented in Section~\ref{sec:theory2}. Finally, some additional heuristic algorithms and experimental results are detailed in Section~\ref{sec:exp}. 

\section{Related Work}
\label{sec:rel}
Various interesting approaches for {\em team formation} have been studied over the years. In operations research~\cite{CL, ZK, BDD, WOMJ}, the problem is defined as finding an optimal match between people and demanded functional requirements. It is often solved using techniques such as simulated annealing, branch-and-cut or genetic algorithms~\cite{BDD, ZK, WOMJ}. Another interesting problem formulation requires taking into consideration the psychological aspects of the individuals involved in order to form a team of efficient collaboration, e.g, the work by Fitzpatrick and Askin ~\cite{FA}, and Chen and Lin in~\cite{CL}. Although all these approaches are interesting, they do not use the possible presence of a social graph structure between the individuals. Therefore, these approaches are complementary to ours. Further, Gaston et al.~\cite{GSJ} provide an experimental study on the effects of a graph structure among individuals on the performance of a team. 

Our problem formulation differs from these fundamentally by requiring a solution where the optimality is determined based on the properties associated with a social graph structure among the individuals. In particular, we aim to form a team that contains at least $k_i$ nodes of skill $i$ such that the density of the resulting solution subgraph is maximized. A similar problem has been addressed by Lappas et. al.~\cite{LLT}. They try to find a team that contains at least $1$ node for each skill $i$, with the cost of a solution measured in terms of either a diameter or a minimum spanning tree. Our problem definition generalizes this requirement and suggests a new density based measure for solution's objective.

The problem of finding size-bound densest subgraphs is well-studied. Finding a maximum density subgraph on an undirected graph can be solved in polynomial time~\cite{G84, L}. However, the problem becomes NP-hard when a size restriction is enforced. In particular, finding a maximum density subgraph of size exactly $k$ is NP-hard~\cite{AHI, FKP} and no approximation scheme exists under a reasonable complexity assumption~\cite{K}.
Khuller and Saha~\cite{KS} considered the problem of finding densest subgraphs with size restrictions and showed that these are NP-hard. Khuller and Saha ~\cite{KS} and also Andersen and Chellapilla ~\cite{AC} gave constant factor approximation algorithms. Our problem definition varies from these because we not only require to find the maximum density subgraph of size at least $k$, but, we also require that this subgraph contain $k_i$ nodes of property (or skill) $i$ such that $k = k_1 + k_2 + ... + k_n$. Thus, we also generalize past work on finding size-bound maximum density subgraphs.

\section{Problem Definition}
\label{sec:def}

\noindent{\bf Notation.}
Let ${\cal X} = \{ 1, \ldots,n \}$ denote a set of $n$ individuals and ${\cal A} =$ \{$a_1, \ldots, a_m$\} denote a set of $m$ skills. Each individual $i$ is associated with a set of skills $X_i \subseteq {\cal A}$. If $a_j \in X_i$, then an individual $i$ has skill $a_j$. For each skill $a$, we define its support set, $S(a)$, as the set of individuals in $\cal X$ with skill $a$. That is, $S (a) = \{ i |  i \in {\cal X}$ and $a \in X_i \}$. A task $\cal T$ is a set of pairs where each pair, <$a_j$,$k_j$>, specifies that at least $k_j$ individuals of skill $a_j$ are required to perform the task.  

Let $G({\cal X}, E)$ denote the undirected, weighted graph representing the social network associated with the set of individuals ${\cal X}$. 
We use the notations $E(G)$ and $V(G)$ to represent the edge set and vertex set associated with the graph $G$. If ${\cal X'} \subseteq V(G)$, we use $G[\cal X']$ to denote the subgraph of $G$ induced by the nodes in $\cal X'$. Further, $W(\cal X')$ denotes the sum of the edge-weights associated with all the edges in the subgraph induced by the nodes in $\cal X'$. We also define a distance function between any two node $i, i'$ in a graph $G$ as the sum of the edge-weights along the shortest path between $i$ and $i'$ in $G$. Further, without loss of generality, we assume that the graph $G$ is connected; we can transform every disconnected subgraph to a connected one by simply adding an edge that denotes zero collaborative compatibility. Given a measure of collaborative compatibility $Cc()$, we now formalize the problems considered in this paper.

\noindent{\bf Single Skill Team Formation (sTF).} Given a set of $n$ individuals ${\cal X} = \{ 1, \ldots, n \}$, a graph $G({\cal X}, E)$, task ${\cal T} = \{<a, k>\}$, find $\cal X' \subseteq X$, such that $|{\cal X}'  \cap S(a)| \ge k$, and the collaborative compatibility $Cc(\cal X')$ is optimized. 

\noindent{\bf Multiple Skill Team Formation (mTF).} Given a set of $n$ individuals ${\cal X} = \{ 1,\ldots, n \}$, a graph $G({\cal X}, E)$, task ${\cal T} =$ \{$<a_1, k_1>, <a_2, k_2>, \ldots, <a_m, k_m>$\}, find $ \cal X' \subseteq X$, such that $|{\cal X}'  \cap S(a_j)| \ge k_j$  for each $j \in \{ 1, \ldots, m \}$ and the collaborative compatibility $Cc(\cal X')$ is optimized.


The main metric that we consider for collaborative compatibility for {\it sTF} and {\it mTF} is the following density based objective. In addition to this, we consider a diameter based objective as well (suggested in~\cite{LLT}) for comparison.

\noindent{\bf Maximum Density(D).} Given a graph $G({\cal X}, E)$ and a set of individuals $\cal X' \subseteq X$, we define the density collaborative compatibility of $\cal X'$, denoted by {\it Cc-D}$(\cal X')$ to be the density of the induced  subgraph $G[\cal X']$. Recall that the density $d(G)$ of a graph $G$ is defined as $d(G) = \frac{W(G)}{|V(G)|}$ . The higher the value of the density, the better is the collaborative compatibility. An optimal solution $\cal X' \subseteq X$, is the team that can perform task $\cal T$ and has maximum density. 

\noindent{\bf Minimum Diameter(R).} Given a graph $G({\cal X}, E)$ and a set of individuals $\cal X' \subseteq X$, we define the diameter collaborative compatibility of $\cal X'$, denoted by {\it Cc-R}$(\cal X')$, to be the diameter of the subgraph $G[\cal X']$. Recall that the diameter of a graph is the largest shortest path between any two nodes in the graph. 
An optimal solution $\cal X' \subseteq X$, is the team that can perform task $\cal T$ and has minimum diameter. 

In the following sections, we refer to the {\it Single Skill Team Formation (sTF)} and {\it Multiple Skill Team Formation (mTF)} problems with collaborative compatibility {\it Cc-R} as {\it Diameter-sTF} and {\it Diameter-mTF}, respectively.  Similarly, for the collaborative compatibility {\it Cc-D} we refer to the corresponding problems as {\it Density-sTF} and {\it Density-mTF} respectively.

\noindent{\bf Properties.} We now describe some properties of the maximum density objective. Notice that neither of these properties hold on {\it Diameter-sTF} or {\it Diameter-mTF}.
For brevity, we mention the intuition without a rigorous definition or proof.

\noindent{\bf Strict Monotonicity.} If a communication edge (with positive weight) is added between two nodes in the solution set for the {\it Density-sTF} or {\it Density-mTF} problem, then the collaborative compatibility objective {\it Cc-D} for the solution necessarily increases. Similarly, if a communication edge already present is deleted, then the {\it Cc-D} objective value decreases. 
This seems intuitive as an added collaboration between two people in the team enhances the quality of the team. However, in the case of diameter, adding or deleting an edge may not affect the solution at all. 

\noindent{\bf Sensitivity.} The {\it Cc-D} value for {\it Density-sTF} or {\it Density-mTF} does not increase or decrease radically upon adding or deleting an edge. Specifically, it can only change to an extent depending on the weight of the added or deleted edge, compared to the total weight of edges in the solution. However, adding or deleting an edge can radically change the diameter (for example make it finite from infinite) for an induced subgraph; this implies that the diameter objective is highly sensitive to small change. 


The properties for density based objectives fall out of the fact that adding or deleting edges only gradually alters the density of a solution subgraph. Diameter based objectives (or even the minimum spanning tree based objective suggested in~\cite{LLT} that we do not consider in this paper) are not smooth in this sense; altering the graph slightly can change the objective radically. These properties make density based objectives somewhat more suitable. One drawback, however, of density as an objective arises from the fact that the optimal solution may contain disconnected components. Notice that this is not the case for the diameter based objective, however, although the solution returned is connected it may be of large size including non-skilled (undesired) nodes that are required to ensure the connectivity. To ensure the connectivity property for the density-based solutions, in the experimental section we suggest several heuristic algorithms.

Eventually, the quality of teams produced by different definitions
needs to be evaluated (potential for collaboration) based on the measures neutral to these definitions; we make such objective comparisons in the experimental section.

\section{Density-based objective}
\label{sec:theory1}
In this section, we claim that {\it Density-sTF} and {\it Density-mTF} are NP-hard problems. We then present the algorithms {\it s-DensestAlk} (Algorithm ~\ref{algo:sDlk})  and {\it m-DensestAlk} (Algorithm ~\ref{algo:mDlk}) for {\it Density-sTF} and {\it Density-mTF}, respectively. Further, we prove that {\it Density-sTF} achieves 3-approximation factor. 

\begin{theorem}
{\it Density-sTF} and {\it Density-mTF} problems are NP-complete.
\end{theorem}
\begin{proof}
We prove the $claim$ by a reduction from the {\it Densest at least $k$ subgraph (DalkS)} problem defined in ~\cite{KS}. An instance of {\it DalkS} consists of a graph $G({\cal X}, E)$, and a constant $k$, and the solution is a maximum density subgraph with at least $k$ nodes. We transform it into an instance of {\it Density-sTF} problem by defining a skill $a$ for every node $v \in V$ in which case a solution would be a maximum density subgraph with at least $k$ nodes that have skill $a$. And since skill $a$ is defined for every node in $G$, it is easy to see that ${\cal X'} \subseteq {\cal X} $ is the solution to the problem {\it Density-sTF}  iff it is a solution to the problem {\it DalkS}. The problem {\it Density-sTF} is a special case of {\it Density-mTF} which implies that {\it Density-mTF} is NP-hard. 
\end{proof}

\subsection{3-approximation algorithm for {\it Density-sTF}}
\label{subsec:density-sTF}
{\it Intuition}: To begin with, the algorithm {\it s-DensestAlk} (Algorithm ~\ref{algo:sDlk}) accepts the graph and the skill requirements as an input. It then finds the densest subgraph and removes it from the input graph and adds it to the solution subgraph (which is initially empty). It then checks if the solution subgraph satisfies the skill  requirements. Until the solution subgraph constructed meets the skill requirements, the algorithm continues to iterate through the process of finding the densest subgraph from the remaining input graph and adding it to the solution subgraph. Since in each iteration the algorithm adds the densest subgraph, it is ensured that the solution subgraph has sufficiently high density. Note that although we are able to prove that the algorithm guarantees a $3$-approximation ratio in terms of density, no bound on the size is guaranteed. We overcome this drawback by applying various simple heuristic algorithms which are described later in the section ~\ref{subsec:heuristic}.

{\it Details}: The algorithm {\it s-DensestAlk(G, {\cal T})}  takes as input the social graph $G$ and a task ${\cal T} = \{$<$a, k$>$\}$ where at least $k$ individuals/nodes of skill $a$ are required to perform the task ${\cal T}$. As explained intuitively, the algorithm then proceeds through multiple iterations. In each iteration, $i$, it finds the maximum density subgraph of $G_i$, say $H_{i+1}$, removes it from $G_i$ using the routine $shrink(G_i, H_{i+1})$ and constructs a new solution subgraph $D_{i+1}$ using the routine $union(D_i, H_{i+1})$. The routine $shrink(G, H)$ removes $H$ from $G$ such that for each $v \in (G - H)$, if $v$ has $l$ edges to the vertices in $H$, then it adds $l$ self-loops to $v$ with the corresponding edge-weights. Inside the routine $union(D, H)$, then for each loop, we look at its corresponding edge, say $e(u, v)$, in the original input graph, $G$, and if $u \in D, v\in H$ (or vice-versa), we replace the loop by an edge $e(u, v)$. Finally, once the loop-termination condition is satisfied, the algorithm then examines each of the intermediate solution subgraphs, $D_i$, constructed in previous iterations and adds sufficient number of skilled nodes to it so that each $D_i$ satisfies the skill requirement. The algorithm then picks up the one with the highest density as the final solution subgraph. 

Our algorithm is very similar to the {\it DensestAtleastK} algorithm in ~\cite{KS} that calculates the maximum density subgraph containing at least $k$ vertices without any skill constraints imposed. The naive extension would be to just add $k$ skilled nodes to the solution returned by algorithm {\it DensestAtleastK}. And since their algorithm guarantees an approximation factor of $2$ for density, this naive extension would guarantee an approximation factor of $4$ (proof omitted for brevity). But, since the additional $k$ nodes are picked at random the solution may suffer from many disconnected components making it practically infeasible to be of any use. Therefore, we propose the  algorithm {\it s-DensestAlk} that differs mainly in the loop-termination condition imposed. This condition ensures that the resulting solution satisfies the constraints of at least $k$ skilled nodes, improves the approximation ratio to $3$ from $4$, and has good connectivity properties.

Although the proof for $4$-approximation is simple, it turns out that proving a 3-approximation to {\it Density-sTF} is significantly harder. While the algorithm is simple, the analysis is fairly detailed. The key idea is to consider various cases about the returned subgraph and carefully examine the density of each component. The analysis is similar to~\cite{KS} at the high level. However, due to the skill-set constraints, several sub-cases need to be considered.
\begin{algorithm}[]
\caption{s-DensestAlk($G, {\cal T}$)}
\label{algo:sDlk}
\begin{algorithmic}[1]
\STATE $D_0  \leftarrow \phi, \ G_0 \leftarrow G, \ i \leftarrow 0$
\WHILE{ $|D_i \cap S(a)| < k$ where ${\cal T} = \{$<$a, k$>\}}
\STATE $H_{i + 1} \leftarrow$ maximum-density-subgraph$(G_i)$
\STATE $D_{i + 1} \leftarrow union(D_i, H_{i + 1})$
\STATE $G_{i + 1} \leftarrow shrink(G_i,  H_{i + 1})$
\STATE $i \leftarrow i + 1$
\ENDWHILE
\FOR {$each \ D_i$}
\STATE $n_a =$ number of nodes of skill $a$ in $D_i$
\STATE Add $max(k - n_a, 0)$ nodes of skill $a$ to $D_i$ to form $D'_i$
\ENDFOR
\STATE Return $D'_i$ which has the maximum density
\end{algorithmic}
\end{algorithm}
\begin{theorem}
The algorithm {\it s-DensestAlk} achieves an approximation factor of 3 for the {\it Density-sTF} problem.
\end{theorem}
\begin{proof}
\setcounter{equation}{0}
Let $H^*$ denote an optimal solution and $d^* =  \frac{W(H^*)}{|V(H^*)|}$ denote density of the optimal solution. 

If the number of iterations is 1, then $H_1$ is the maximum density subgraph that contains at least $k$ nodes of skill $a$. Therefore, $H^* = H_1$ and the algorithm returns it. Otherwise, say the algorithm iterates for $l \ge 2$ rounds. There can be two cases:

\noindent {\bf Case 1:} There exists an $l' < l$ such that \\
$W(D_{l' - 1} \cap H^*) < \frac{W(H^*)}{2}$  and  $W(D_{l'} \cap H^*) \ge \frac{W(H^*)}{2}$.

\begin{figure}[t]
\begin{center}
\includegraphics[scale=0.25, bb = 30 300 500 700]{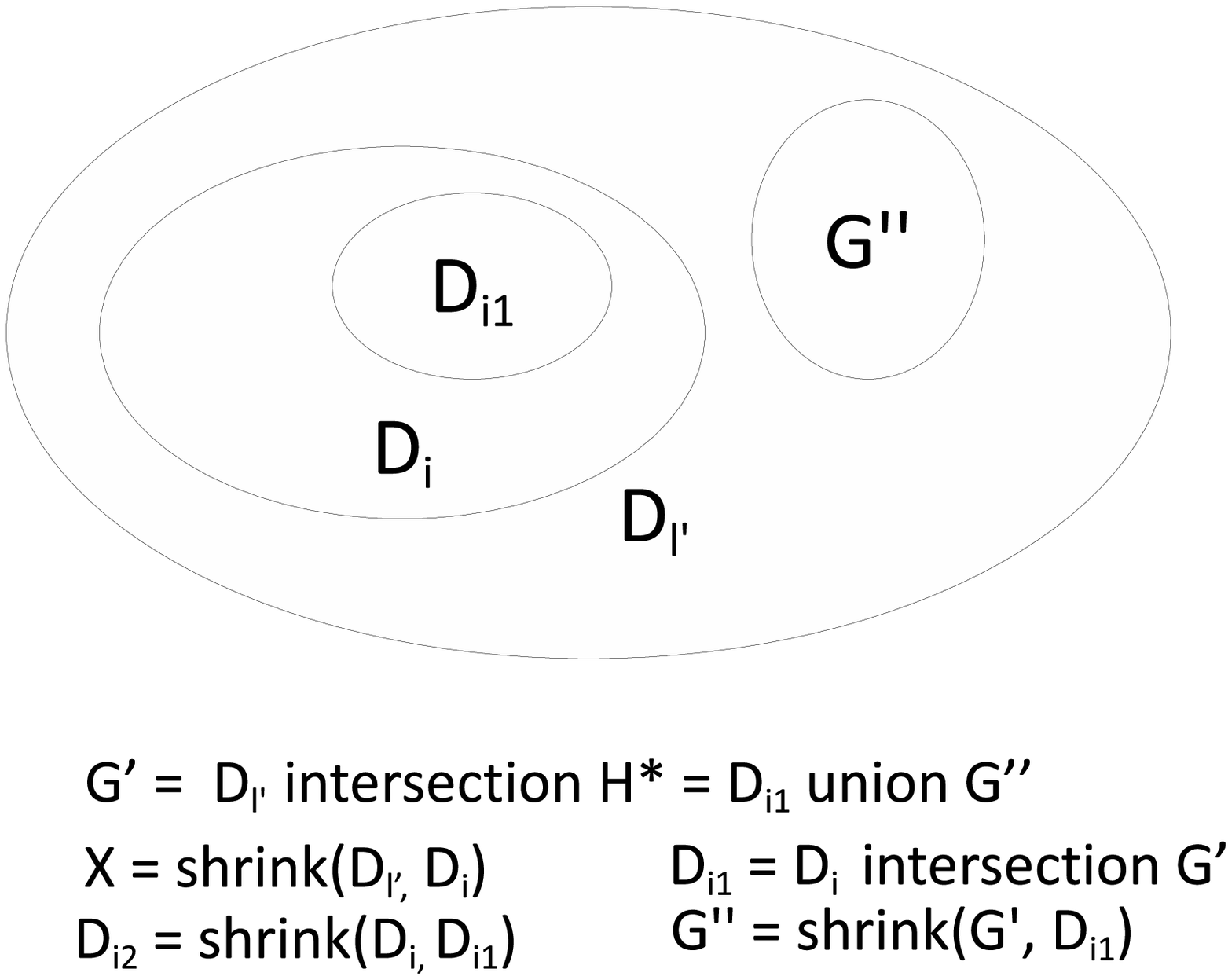}
\caption{$D_{l'} = D_{i1} \cup D_{i2} \cup X$}\label{fig:3approx}
\end{center}
\end{figure}

\noindent {\bf Case 2:} There exists no such $l' < l$.

Before analyzing the two cases in detail, note that by construction
$density(H_i) \le density(D_i) \le density(D_{i - 1})$. We now consider case 2 first and later case $1$.

\noindent {\bf Proof for Case 2.} 

Since the algorithm terminates after $l$ iterations, $D_l$ contains at least $k$ nodes of skill $a$. Further, we know that for each $j \le l - 1, W(D_j \cap  H^*)  < \frac{W(H^*)}{2}$ \\ 
$\Rightarrow W(G_j \cap  H^*)  \ge \frac{W(H^*)}{2}$ \\ 
$\Rightarrow \frac{W(G_j \cap H^*)}{|V(G_j \cap H^*)|} \ge \frac{W(H^*)}{2 |V(H^*)|}$ \\ 
$\Rightarrow G_j$ contains a subgraph of density $\ge \frac{d^*}{2}$ \\ 
$\Rightarrow density(H_l) \ge \frac{d^*}{2}$ \\ 
$\Rightarrow density(D_l) \ge \frac{d^*}{2}$

Thus, $D_l$ has density $\ge \frac{d^*}{2}$ and contains at least $k$ nodes of skill $a$. Therefore, the algorithm indeed returns a subgraph of density at least $\ge \frac{d^*}{2}$.

\noindent{\bf Proof for Case 1}

\noindent $W(D_{l' - 1} \cap H^*) < \frac{W(H^*)}{2}$  and  $W(D_{l'} \cap H^*) \ge \frac{W(H^*)}{2}$\\
$\Rightarrow W(G_{l'} \cap H^*) \ge \frac{W(H^*)}{2}$ where $G_{l'} = shrink(G, D_{l' -1})$ \\ 
$\Rightarrow \frac{W(G_{l'} \cap H^*)}{|V(G_{l'} \cap H^*)|} \ge \frac{W(H^*)}{2 |V(H^*)|} = \frac{d^*}{2}$ \\ 
$\Rightarrow G_{l'}$ has a subgraph of density $\ge \frac{d^*}{2}$ \\ 
$\Rightarrow density(H_{l'}) \ge \frac{d^*}{2}$  ($H_{l'}  \mbox { is densest subgraph of } G$) \\ 
$\Rightarrow density(D_{l'}) \ge \frac{d^*}{2}$

Now, let us divide {\bf Case 1} into following $4$ parts 
\begin{enumerate}[(a)]
\item $|V(D_{l'})| \le  \frac{k}{2}$ \\
According to step $10$, algorithm adds at most $k$ vertices to $D_{l'}$ to obtain the subgraph, say $D$, with density $d$\\
$d \ge \frac{W(D_{l'})}{|V(D_{l'})| + k} \ge \frac{\frac{W(H^*)}{2}}{\frac{k}{2} + k} \ge \frac{\frac{W(H^*)}{2}}{\frac{3 |V(H^*)|}{2}} = \frac{d^*}{3}$ 

\item $|V(D_{l'})| \ge  2k $

According to step $10$, algorithm adds at most $k$ vertices to $D_{l'}$. Further, we know that $density(D_{l'}) \ge \frac{d^*}{2}$ therefore, the resulting subgraph, $D'_{l'}$ has density 

\noindent $d = \frac{W(D_{l'})}{|V(D_{l'})| + k} \ge \frac{W(D_{l'})}{\frac{3}{2} |V(D_{l'})|}  \ge  \frac{d^*}{3}$

\item $\frac{k}{2} < |V(D_{l'})| <  2k$ and $ |V(D_{l'}) \cap V(H^*)| \ge \frac{|V(H^*)|}{2}$

According to step $10$, algorithm adds at most $\frac{|V(H^*)|}{2}$ nodes to $D_{l'}$ to form $D'_{l'}$ with density, say $d$.

\begin{enumerate}[i]
\item $|V(D_{l'})| \ge |V(H^*)|$ \\
$d \ge \frac{W(D_{l'})}{|V(D_{l'})| + \frac{|V(H^*)|}{2}} \ge \frac{W(D_{l'})}{|V(D_{l'})| + \frac{|V(D_{l'})|}{2}}  \ge \frac{d^*}{3}$ 

\item $|V(D_{l'})| < |V(H^*)|$ \\
$d \ge \frac{W(D_{l'})}{|V(D_{l'}|) + \frac{|V(H^*)|}{2}} \ge \frac{W(D_{l'})}{|V(H^*)| + \frac{|V(H^*)|}{2}} \ge  \frac{\frac{W(H^*)}{2}}{\frac{3}{2}|V(H^*)|} \ge \frac{d^*}{3}$
\end{enumerate} 

\item $\frac{k}{2} < |V(D_{l'})| <  2k$ and $|V(D_{l'}) \cap V(H^*)| < \frac{|V(H^*)|}{2}$

If $d_{l'} = density(D_{l'}) \ge d^*$, then adding at most $k$ vertices gives a subgraph $D'_{l'}$ with density, say $d$ such that

\noindent $d = \frac{W(D_{l'})}{ |V(D_{l'})| + k} \ge \frac{W(D_{l'})}{|V(D_{l'})| + 2 |V(D_{l'})|}  \ge \frac{W(D_{l'})}{3 |V(D_{l'})|} \ge \frac{d^*}{3}$

Therefore, $D_{l'}$ is a subgraph that contains at least $k$ nodes of skill $a$ and has density $d \ge \frac{d^*}{3}$. We are done here.

Now, assume that $d_{l'} < d^*$.

In the rest of the proof, we divide $D_{l'}$ into subgraphs as explained below and shown in Figure~\ref{fig:3approx}. 

$\mbox{Let } G' = D_{l'} \cap H^* $. 

\begin{claim}
\label{claim:case2Claim1}$W(G')\ge \frac{W(H^*)}{2}$ and $density(G')\ge d^*$.
\end{claim}
\begin{proof}
$|V(G')| = |V(D_{l'} \cap H^*)| < \frac{|V(H^*)|}{2} $ and $W(G') = W(D_{l'} \cap H^*) \ge  \frac{W(H^*)}{2} $. \\
$\Rightarrow density(G') \ge \frac{\frac{W(H^*)}{2}}{\frac{|V(H^*)|}{2}} \ge d^* $.
\end{proof}

Define $i$ such that $density(H_i)\ge d^*$ and $density(H_{i+1}) < d^*$. Such an $i\le l'$ exists due to {\sc Claim ~\ref{claim:case2Claim1}} and since $d_{l'} < d^*$.\\ 

$\Rightarrow density(D_i)  = d_i \ge d^*$. 

Let, $n_i = |V(D_i)| $. We now consider two sub-cases. 

\begin{enumerate}[i]
\item $n_i \ge \frac{|V(H^*)|}{2}$: Add at most $k$ vertices to $D_i$ to get a subgraph $D'_i$ with $density(D'_i) = d$, such that  \\
$d = \frac{W(D_i)}{|V(D_i)| + k} \ge \frac{W(D_i)}{|V(D_i)| + |V(H^*)|} \ge \frac{W(D_i)}{3 |V(D_i)|} \ge \frac{d^*}{3}$. \\ 
Thus, $D'_i$ is a subgraph containing at least $k$ nodes of skill $a$ and density $d \ge  \frac{d^*}{3}$ and we are done here.

\item $n_i < \frac{|V(H^*)|}{2}$:  
We know that $density(G') \ge  d^*$, $density(H_i) \ge d^*$ and $density(H_{i+1}) < d^*$.  Therefore, $G' \cap D_i \ne \phi$.
We now introduce a few definitions and prove claims about them. 

Let, $D_{i1} = D_i \cap G'$, $D_{i2} = shrink(D_i, D_{i1})$, and $G'' = shrink(G', D_{i1})$ (Figure: ~\ref{fig:3approx1}). Further, let $X = shrink(D_{l'}, D_i)$. \\

\begin{figure}[t]
\begin{center}
\includegraphics[scale=0.5, bb = 30 300 500 500]{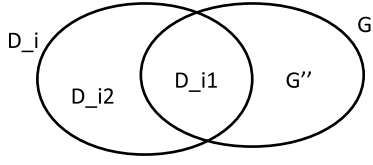} 
\caption{$D_{i1} = D_i \cap G'$, $D_{i2} = shrink(D_i, D_{i1})$, $G'' = shrink(G', D_{i1})$}\label{fig:3approx1}
\end{center}
\end{figure}


\begin{claim} 
\label{claim:case2Claim2}$W(D_{i1}) \ge  \frac{|V(H^*)| d^*}{2} - W(G'')$. 
\end{claim}
\begin{proof}
$W(G') = W(G'') + W(D_{i1})$ since $G'' = shrink(G', D_{i1})$; but $W(G')\ge \frac{W(H^*)}{2}$ (using {\sc Claim ~\ref{claim:case2Claim1}})
\end{proof}

\begin{claim}
\label{claim:case2Claim3} $density(D_{i2}) > \frac{d^*}{2}$.
\end{claim}
\begin{proof} 
Recall that for each $j \le i, density(H_j) > d^*$. Further, $H_j = \mbox{ densest subgraph of } shrink(G, D_{j-1})$.
Therefore, for each $v\in H_j$, the degree of $v$ induced in $H_j$ is at least $d^*$. Therefore, for all $v\in D_i$, $degree(v) > d^*$ (here we abuse notation to denote $v$'s degree induced in $D_{i}$ by $degree(v)$).
\end{proof}

For convenience, let $n_{x} = |V(X)|$, $n_{l'} = |V(D_{l'})|$, $n_{i1} = |V(D_{i1})|$, $n_{i2} = |V(D_{i2})|$, and $n'' = |V(G'')|$.

\begin{claim}
\label{claim:case2Claim4}$W(X) - W(G'') \ge \frac{d^*}{2} (n_{x} - n'')$.
\end{claim}
\begin{proof}
Since $H_{l'}$ is the maximum density subgraph of $shrink(G, D_{l'-1})$, $density(H_{l'}) \ge density(S)$ for any $S \subseteq H_{l'},$. 
Further, since $X = shrink(D_{l'}, D_i)$, and $G'' = shrink(G', D_i \cap G')$, we have  $G'' \subseteq X$.
Therefore, $density(H_j) \ge density(H_j \cap G'')$ (for all $i < j \le l'$). \\
Therefore, $W(X) - W(G'') $\\
$= \sum_{j=i+1}^{l'}{W(H_j)} - \sum_{j=i+1}^{l'}{W(H_{j} \cap G'')}$\\
$\ge \sum_{j=i+1}^{l'}{density(H_{j}) (\mid H_{j} \mid - \mid H_{j} \cap G'' \mid)}$\\
$\ge \frac{d^*}{2} (n_x - n'')$.
\end{proof}

Notice that we have (lower) bounded the density or the weight of each of $D_{i1}$, $D_{i2}$, and $X$, the three components that add up to $D_{l'}$. We are now ready to argue about the density of $D_{l'}$ when $k$ vertices are added to it. Before initiating this analysis, we briefly state a claim relating the sizes of these components. 

\begin{claim}
\label{claim:case2Claim5}$n_{i2} + n_x - n''\ge n_{l'} - \frac{|V(H^*)|}{2}$
\end{claim}
\begin{proof}
This follows using $|V(G')|  \le \frac{|V(H^*)|}{2}$ and the definition $G'' = shrink(G', D_{i1})$.
\end{proof}

We now complete the analysis. \\  

\noindent $d = density(D) \ge \frac{W(D_{l'})}{n_{l'} + k}$ \\
$= \frac{W(D_i) + W(X)}{n_{l'} + k} = \frac{W(D_{i1}) + W(D_{i2}) + W(X)}{n_{l'} + k}$ \\ 
$\ge \frac{\frac{d^* |V(H^*)|}{2} - W(G'') + \frac{d^*n_{i2}}{2} + W(X)}{n_{l'} + k}$ (using {\sc Claim ~\ref{claim:case2Claim2},~\ref{claim:case2Claim3}})\\ 
$\ge \frac{\frac{d^* |V(H^*)|}{2} + \frac{d^*n_{i2}}{2} + \frac{d^*}{2}(n_x - n'')}{n _{l'}+ k}$ (using {\sc Claim ~\ref{claim:case2Claim4}})\\ 
$\ge \frac{d^*}{2} \frac{|V(H^*)| + n_{l'} - \frac{|V(H^*)|}{2}}{n_{l'} + k}$ (using {\sc Claim ~\ref{claim:case2Claim5}}) \\ 
$\ge \frac{d^*}{4} \frac{2n_{l'} + k}{n_{l'} + k} \ge \frac{d^*}{3}$ (since $\frac{k}{2} < n_{l'}$).
\end{enumerate}
\end{enumerate} 
Remark: Cases (c) and (d) do not use the bound $|V(D_{l'}) < 2k|$; so they together subsume case (b), but we have presented (b) for clarity.
\end{proof}

\subsection{Algorithm for {\it Density-mTF}}
In this section, we present the algorithm {\it m-DensestAlk} (Algorithm ~\ref{algo:mDlk}) for the {\it Density-mTF} problem. This is an extension of the algorithm {\it s-DensestAlk} for the {\it Density-sTF} problem described earlier. The algorithm {\it m-DensestAlk} accepts input parameters: graph $G$ and  task ${\cal T} = \{<a_1, k_1>, <a_2, k_2>, \ldots, <a_m, k_m>\}$  which requires at least $k_i$ individuals of skill $a_i$ to perform the task $\cal T$. Each iteration within the algorithm {\it m-DensestAlk} is exactly similar to the {\it s-DensestAlk} described earlier except that here the termination condition verifies that the solution subgraph contains at least $k_i$ nodes with skill $a_i$ for $i \in \{ 1 \cdot \cdot m \}$ and thus satisfying the {\it multiple} skill requirement instead of {\it single} skill requirement. The details of the algorithm are similar to that described for {\it s-DensestAlk} in the section ~\ref{subsec:density-sTF}.
\begin{algorithm}[]
\caption{m-DensestAlk($G, {\cal T}$)}
\label{algo:mDlk}
\begin{algorithmic}[1]
\STATE $D_0  \leftarrow \phi, \ G_0 \leftarrow G, \ i \leftarrow 0$
\WHILE{ $|D_i \cap S(a_j)| < k_j$ for any $<a_j, k_j> \in {\cal T}$}
\STATE $H_{i + 1} \leftarrow$ maximum-density-subgraph$(G_i)$
\STATE $D_{i + 1} \leftarrow union(D_i, H_{i + 1})$
\STATE $G_{i + 1} \leftarrow shrink(G_i,  H_{i + 1})$
\STATE $i \leftarrow i + 1$
\ENDWHILE
\FOR {$each \ D_i$}
\STATE $D'_i \leftarrow D_i$
\FOR {$each \ <a_1, k_1> \in {\cal T}$}
\STATE $n_{aj} =$ number of nodes of skill $a_j$ in $D_i$
\STATE Add $max(k_j - n_{aj}, 0)$ nodes of skill $a_j$ to $D'_i$
\ENDFOR
\ENDFOR
\STATE Return $D'_i$ which has the maximum density
\end{algorithmic}
\end{algorithm}
\begin{theorem}
The algorithm {\it m-DensestAlk} achieves an approximation factor of 3 for the special case of {\it Density-mTF} problem where each node in the graph has at most one skill.
\end{theorem}
\begin{proof}
Let $m =\ \mid${\cal T}$\mid$ and $k = \sum_{j=1}^{m}{k_j}$ where $k_j$ number of individuals are required of skill $a_j$ s.t. $<a_j,k_j> \in {\cal T}$. Since each node contributes to atmost one skill, an optimal solution, $H^*$, has at least $k$ vertices. The proof for {\it m-DensestAlk} is analogous to the proof for {\it s-DensestAlk} with the only difference that instead of adding any $k$ nodes of skill $a$ to $D_i$s, we add $k_j$ nodes of skill $a_j$ s.t. $<a_j,k_j> \in {\cal T}$.
\end{proof}

We are unable to bound the performance of {\it m-DensestAlk} for the general case of {\it Density-mTF} problem. Futher, the time complexity of {\it m-DensestAlk} is $O(k n^3)$ which can be inefficient for very large graphs but is manageable at the scale at which we run experiments. Directly using the linear time algorithm for the densest at least $k$ subgraph problem in~\cite{KS,AC} or $O(n^3)$-time algorithm from~\cite{KS,AC} for {\it Density-sTF} problem would result in a weaker bound i.e. $6$ and $4$-approximation respectively. In both cases, however, one may possibly get many disconnected components.


\section{Diameter-based objective}
\label{sec:theory2}
In this section, we mention theoretical results for {\it Diameter-sTF} and {\it Diameter-mTF}. We show that these problems are NP-hard (note that the NP-hardness of {\it Diameter-sTF} does not follow from any previous work). We further present an algorithm {\it MinDiameter} (Algorithm ~\ref{algo:minDia})  which is an extension of {\it RarestFirst} in~\cite{LLT}, and prove that it achieves a 2-approximation factor.

\begin{algorithm}[]\label{algo:minDia}
\caption{MinDiameter(G, {\cal T})} 
\label{algo:minDia}
\begin{algorithmic}[1]
\FOR{each $<a, k> \in T$}
\STATE $S(a) = \{ i \mid a \in X_i \}$
\ENDFOR
\STATE $a_{rare} = \arg \min_{<a, k>\ \in T} |S(a)|$
\FOR{each $i \in S(a_{rare})$}
\FOR{each $<a, k>\ \in T$}
\STATE $R_{ia} = d_k(i, S(a), k)$
\ENDFOR
\STATE $R_i \leftarrow \max_a R_{ia}$
\ENDFOR
\STATE $i^* \leftarrow \arg \min R_i$
\STATE ${\cal X'} = \{ i^* \}$
\FOR{each $<a, k>\ \in T$}
\STATE ${\cal X'} = {\cal X'} \cup \{  Path_k(i^*, S(a), k) \}$
\ENDFOR
\end{algorithmic}
\end{algorithm}

\begin{theorem} 
{\it Diameter-sTF} and {\it Diameter-mTF} problems are NP-complete.
\end{theorem}
\begin{proof}
The problems {\it Diameter-sTF} and {\it Diameter-mTF} are in NP because for a given candidate solution, in polynomial time, it can be verified that the skill-set requirement is satisfied. We prove that {\it Diameter-sTF} is NP-hard by reduction from the 3-satisfiability problem.  Consider a 3-SAT instance, say $\Psi = C_1 \wedge C_2 ... \wedge C_m$, where each clause, $C_j = (x \vee y \vee z)$, and $\lbrace x, y, z \rbrace \in U = \lbrace u_1, \neg u_1, u_2, \neg u_2, \cdot \cdot \cdot, u_n, \neg u_n \rbrace$. Let, $C = \lbrace  C_1, C_2, \cdot \cdot \cdot, C_m \rbrace$. Let $N, M$ denote the number of variables and clauses, respectively. We construct an instance of {\it Diameter-sTF} problem corresponding to the 3-SAT instance $\Psi$ using the following rules.

{\it Rule $1$} For each variable $x$, create two nodes $x, \neg x$ in $G$ and set $w(x, \neg x) = r'$.

{\it Rule $2$} For each clause $C_j$, create two nodes, $C_{j1}$ and $C_{j2}$ in $G$ and set $w(C_{j1}, C_{j2}) = r'$. 

{\it Rule $3$} Pick any $r$ such that $r < r'$. For each pair of variables ($x, y$) where $y \ne \neg x$, set $w(x, y) = r$. Similary, for each pair of clauses ($C_f, C_g$), where $w(C_f, C_g)$ is not set by rule $2$, set $w(C_f, C_g) = r$.

{\it Rule $4$} 
For each clause, $C_j = (x \vee y \vee z)$, set 

$w(C_{j1}, x) = w(C_{j1}, y) = w(C_{j1}, z) = \frac{r}{2}$ and 

$w(C_{j2}, x) = w(C_{j2}, y) = w(C_{j2}, z) = \frac{r}{2}$ and 

$w(C_{j1}, u) = w(C_{j2}, u) = r$ for each $u \in U - \{x, y, z\}$ 

{\it Rule $5$} For each $u_i, \neg u_i \in U$, associate a skill $a$ to node $u_i, \neg u_i$. And for each $C_j \in C$, associate a skill $a$ to the nodes $C_{j1}, C_{j2}$. 

\begin{claim}
\label{claim:minDiaClaim1} In $G$, $d(x, \neg x) > r$ where $x, \neg x \in U$.
\end{claim}
\begin{proof}
In $G$, for each variable $y (\ne x \ne \neg x), \ d(x, y) = d(\neg x, y) = r$ and $w(x, \neg x) = r' > r$ (rule $1, 3$). Further, both $x$ and $\neg x$ cannot appear together in any clause $C_j \in C$ (pre-processing). Therefore, in $G,\ d(C_{j1}, x) = d(C_{j2}, x) = \frac{r}{2}$ and $d(C_{j1}, \neg x) = d(C_{j2}, \neg x) = r$ (rule $3, 4$). \\
$\Rightarrow d(x, \neg x) > r$
\end{proof}
\begin{claim}
\label{claim:minDiaClaim2} Let $X$ be the subgraph of $G$ and $V(X)$ denote the nodes in $X$. Let $C_{j1}, C_{j2} \in V(X)$ where $C_j = (x \vee y \vee z)$. Then, in $X$, $d(C_{j1}, C_{j2}) = r$ iff $V(X) \cap \{x, y ,z\} \ne \phi$.
\end{claim}
\begin{proof}
Assume $V(X) \cap \{x, y, z\} = \phi$. \\
In $G$, for each clause $C_f (\ne C_{j1} \ne C_{j2}), \ d(C_{j1}, C_f) = d(C_{j2}, C_f) = r$ and $w(C_{j1}, C_{j2}) = r' > r$ (rule $2, 3$). Further, for each $u \in U - \{x, y , z\},\ d(C_{j1}, u) = d(C_{j2}, u) = r$ (rule $4$). Therefore, in $X,  d(C_{j1}, C_{j2}) > r$. However, this is a contradiction because, in $X, d(C_{j1}, C_{j2}) = r$. \\
$\Rightarrow V(X) \cap \{x, y ,z\} \ne \phi$.
\end{proof}
\begin{claim}
\label{claim:minDiaClaim3} Let  $k = N + 2M$. If $\Psi$ has a satisfying assignment then $G$ has a sub-graph $\cal X'$ with $|{\cal X'} \cap S(a)| \ge k$ and $diameter({\cal X'}) \le r$. 
\end{claim}
\begin{proof}
If $\Psi$ has a satisfying assignment, then $G$ has a subgraph $\cal X'$ such that $\cal X'$ contains $C_{j1}, C_{j2}$ for each clause $C_j \in C$, and $u (or \ \neg u) \in U$ if $u (or \  \neg u)$ is set to $1$ in the satisfying assignment for $\Psi$. Note that in the satisfying assignment for $\Psi$ either $u$ or $\neg u$ appears in the assignment. Thus, $\cal X'$ contains exactly $N$ variables and twice the number of clauses. Thus, $|{\cal X'} \cap S(a)| = N + 2M = k$ (rule $5$). \\
Since $\cal X'$ contains either a variable or it negation, for each pair of variables $(x, y) \in V(X) \cap U,\ d(x, y) = r$ (rule $3$). Further, in the satisfying assignment for $\Psi$ each clause $C_j = x \vee y \vee z$, has at least one of the variables set to $1$. So, for each pair of nodes $(p, q) \in V(X) \cap C, \ d(p, q) = r$ ({\sc Claim ~\ref{claim:minDiaClaim2} } and rule $3$) . Therefore, distance between any two nodes in $\cal X'$ is $r$ (rule $4$).\\
Thus, if $\Psi$ has a satisfying assignment then $G$ has a subgraph $\cal X'$ with $|{\cal X'} \cap S(a)| = k$ and $diameter({\cal X'}) = r$. 
\end{proof}
\begin{claim}
\label{claim:minDiaClaim4} Let  $k =  N + 2M$. If $G$ has a sub-graph $\cal X'$ with $|{\cal X'} \cap S(a)| \ge k$ and $diameter({\cal X'}) \le r$ then $\Psi$ has a satisfying assignment.
\end{claim}
\begin{proof}
If $diameter({\cal X'}) \le r$ then it contains either $u$ or $\neg u$ but not both because $d(u, \neg u) > r$ ({\sc Claim ~\ref{claim:minDiaClaim1}}). Since $k = N + 2M$, for each variable $u \in U$, $\cal X'$ contains a node corresponding to either $u$ or $\neg u$ (not both) and for each clause $C_j \in C$, $\cal X'$ contains nodes corresponding to $C_{j1}$ and $C_{j2}$ (rule $5$). Now, since $diameter({\cal X'}) \le r$, it implies that $d(C_{j1}, C_{j2}) \le r$. This implies that at least one of the nodes corresponding to $x, y, z$ in $C_j$ is included in the sub-graph $\cal X'$ ({\sc Claim ~\ref{claim:minDiaClaim2}}). Now, if each variable $u \in U \cap V({\cal X'})$ is set to $1$, then $\Psi$ has a satisfying assignment.
\end{proof}
{\sc Claims ~\ref{claim:minDiaClaim3}} and ~\ref{claim:minDiaClaim4} prove that {\it Diameter-sTF} is NP-hard. Since {\it Diameter-sTF} is the special case of {\it Diameter-mTF}, its NP-hardness proof follows.
\end{proof}

\begin{theorem}
For any graph distance function $d$ that satisfies the triangle inequality, the algorithm {\it MinDiameter} achieves an approximation factor of $2$ for the {\it Diameter-sTF} and {\it Diameter-mTF} problems.
\end{theorem}
\begin{proof}
The analysis we present here is similar to the analysis of the {\it RarestFirst} algorithm presented in ~\cite{LLT}. First, consider the solution $\cal X'$ output by the {\it MinDiameter} algorithm, and let $a_{rare} \in T$ be the skill possessed by the least number of individuals in $\cal X$. Also, let $i^*$ be the individual picked from set $S(a_{rare})$ to be included in the solution $\cal X'$. Now, consider two other skills $a_1 \ne a_2 \ne a_{rare}$ and individuals $i, i' \in \cal X$ such that $i \in S(a_1), i \not \in S(a_2)$ and $i' \not \in S(a_1), i' \in S(a_2)$. If $i, i'$ are part of the team reported by the {\it MinDiameter} algorithm, it means that $d(i^*, i) \le d_k(i^*, S(a_1), k_1)$ and $d(i^*, i') \le d_k(i^*, S(a_2), k_2)$. Due to the way the algorithm operates, we can lowerbound the Cc-R cost of the optimal solution, $\cal X^*$,  as follows: 

\begin{equation}\label{DiaBounds}
d(i^*, i) \le \mbox{Cc-R}({\cal X^*}) \mbox{ and }  d(i^*, i') \le \mbox{Cc-R}({\cal X^*})
\end{equation}
Since we have assumed that the distance function $d$ satisfies the triangle inequality, \\
$d (i, i') \le d(i, i^*) + d(i^*, i')$\\
By applying the bounds given in ~\eqref{DiaBounds}, we get the proposed approximation factor. \\
$d (i, i') \le$ Cc-R($\cal X^*) +$ Cc-R(${\cal X^*}) \le 2 \cdot $Cc-R($\cal X^*$).
\end{proof}

Algorithm {\it MinDiameter} is as follows. For each individual, say $i_r \in S(a_{rare})$ where $a_{rare}$ is the rarest skill (the skill with the minimum size support set $S$), and for each skill $a_i \in {\cal T}$, the algorithm finds the distance to all the nodes in the support set $S(a_i)$. Then, for each support set $S(a_i)$, it chooses the $k_i$-size subset of $S(a_i)$ such that the maximum shortest path distance between $i_r$ and the nodes in this subset is minimum among all $k_i$-size subsets of $S(a_i)$. We call this distance as $k_i$-th shortest distance between $i_r$ and $S(a_i)$ and denote it as $d_k(i_r, S(a_i), k_i)$. Further, we denote the set of $k_i$ shortest paths between $i_r$ and each of the nodes belonging to the corresponding $k_i$-size subset of $S(a_i)$ as $Path_k(i_r, S(a_i), k_i)$. Thus, for each $i_r \in S(a_{rare})$ the algorithm has identified $k_i$ nodes of skill $a_i$, thereby forming a possible solution team that satisfies the constraints. Finally, the algorithm then picks one of these solutions that has minimum diameter. 
The time complexity of the algorithm {\it MinDiameter}, assuming that all pairs shortest paths are pre-computed, is $O(n^2)$.

\section{Experiments}
\label{sec:exp}
In this section, we evaluate various team formation algorithms using the collaboration graph extracted from the DBLP bibliography server. We show that the density of the subgraph returned by our algorithms {\it s-DensestAlk} and {\it m-DensestAlk} perform favorably in comparison to the algorithm {\it MinDiameter}. We also show that our algorithm for density version provides high-quality results in terms of effective communication and collaboration (according to several metrics). In this section, we also present three simple heuristic extensions that can be used to process the solutions returned by {\it s-DensestAlk} and {\it m-DensestAlk} in order to further improve these solutions by reducing size and improving connectivity, while maintaining high density. 
Finally, examples of teams reported by our methods qualitatively corroborate the effectiveness of our framework.

\subsection{Experimental Setup}
We use a snapshot of the DBLP data downloaded on May 17, 2010 to create a benchmark data set for our experiments. We only consider the papers published in the domains of Database (DB), Data Mining (DM), Artificial Intelligence (AI) and Theory (T) conferences. We select papers from a total of $21$ conferences categorized as follows: $DB = \{\textsc{sigmod, vldb, icde, icdt, edbt, pods}\} $, $DM = \{\textsc{www, kdd, sdm, pkdd, icdm}\}$, $AI =$ \{\textsc{icml, ecml,colt, uai}\}, and $T = \{\textsc{soda, focs, stoc, stacs, icalp, esa}\}$.
We define the skill set $\cal T = \{\textsc{t, ai, db, dm}\}$. The set of skilled individuals $X_{dblp}$ consists of the set of authors with at least three papers in these domains. Two authors $i_1, i_2$ are connected in the graph $G_{dblp} (X_{dblp}, E)$ if they appear as co-authors in at least two papers in DBLP. The above procedure creates a set $X_{dblp}$ consisting of $6137$ individuals. The maximum component size is $3869$. We use this for all the experiments. The skill set $X_i$ of each such author $i$ is defined as $X_i = \{ t \mid t \in {\cal T} \ and \ P_i(t) \ne \phi \}$ where $P_i(t)$ denotes the set of papers coauthored by $i$ that are published in the conferences in the domain $t$. 

{\it Maximum Density Team Formation.} To evaluate the algorithms ~\ref{algo:sDlk} and ~\ref{algo:mDlk}, for each edge $e(i_1, i_2)$, we set the edge weight $w(i_1, i_2) = |P_{i1} \cap P_{i2}|$, where $P_{i1}$ and $P_{i2}$ represent the set of papers published by $i_1$ and $i_2$ respectively. For the subgraph, say $G'(V', E')$ returned by these algorithms, we calculate the density, $d' = \frac{W(G')}{|V(G')|}$. 

{\it Minimum Diameter Team Formation.} Here, we set edge-weight $w(i_1, i_2) = 1 - \frac{|P_{i1} \cap P_{i2} |}{| P_{i1} \cup P_{i2} |}$ as suggested in the paper~\cite{LLT}.  For comparison, when a subgraph $G'(V', E')$ is returned by the {\it MinDiameter}, we compute its density by considering the induced subgraph on vertices $V''$, say $G''$ (which could contain more edges that $E'$). The density calculated is $d'' = \frac{W(G'')}{|V(G'')|}$ with edge weights $w(i_1, i_2) =  |P_{i1} \cap P_{i2}|$. 


\subsection{Heuristic algorithms}
\label{subsec:heuristic}
The objectives for {\it sTF-Density} and {\it mTF-Density} are to find subgraphs with maximum density satisfying the skill requirements. 
However, this does not necessitate a connected graph; disconnectedness makes meaningful collaboration in real-life difficult. This is an artifact of the objective function, rather than the algorithm. While the solutions returned by our algorithms {\it sTF-Density} and {\it mTF-Density} never had more than three components, we would like solutions with only one component. 
This is the motivation for heuristic improvements. A dual benefit in our suggested heuristics is that we are able to reduce the number of nodes in the returned subgraph.
The hope is that these can be achieved without compromising significantly on the density.

\begin{algorithm}[]
\caption{EnhanceComponent($G', T$)}
\label{algo:ecDlk} 
\begin{algorithmic}[1]
\STATE (Note: $T= \{<a, k>\}$)
\FOR {each component $C_i \in G'$}
\STATE $C'_i \leftarrow C_i$, $Ni \leftarrow N(C_i) - C_i$ 
\STATE (note: $N(C_i)$ denotes neighbors of nodes in $C_i$)
\FOR {each node $v \in N_i$}
\IF { $| V(C'_i) \cap S(a) | \ge k$}
\STATE ${\cal C'} \leftarrow {\cal C'} \cup C'_i$
\STATE break for loop
\ENDIF
\IF {$v \in S(a)$}
\STATE $C'_i \leftarrow C'_i \cup v $
\ENDIF
\ENDFOR
\ENDFOR
\end{algorithmic}
\end{algorithm}
\begin{algorithm}[]
\caption{EnhancedDense($G, T$)}
\label{algo:edDlk}
\begin{algorithmic}[1]
\STATE $G' \leftarrow$  {\it s-DensestAlk(G, T)}
\STATE ${\cal C'} \leftarrow {\it EnhanceComponent}(G', T)$
\STATE Return $\arg \min_{C'_i \in {\cal C'}} | C'_i |$ 
\end{algorithmic}
\end{algorithm}
We present three heuristics. The starting point of each is the solution to {\it sTF-Density} or {\it mTF-Density}, as the case may be.
We name these heuristics as {\it EnhancedDense} (Algorithm ~\ref{algo:edDlk}), {\it PartialTrimmedDense} (Algorithm ~\ref{algo:ptDlk})  and {\it CompleteTrimmedDense (Algorithm ~\ref{algo:ctDlk})}. For simplicity in presentation, the algorithms are presented as extensions to {\it s-DensestAlk}, but they apply to {\it m-DensestAlk} analogously. The basic idea behind algorithm {\it EnhancedDense} is to inspect each individual component in the solution and attempt to modify it so that it itself satisfies the skill set requirement imposed by the task $\cal T$. This is done by examining the neighbors of the nodes in the component and adding those neighbors that are skilled nodes. The heuristics {\it PartialTrimmedDense} and {\it CompleteTrimmedDense}, take as an input the components generated by the algorithm {\it EnhanceComponent} (Algorithm ~\ref{algo:ecDlk})  and attempt to reduce the size of each component by removing the non-skilled nodes one by one without making the component disconnected. The {\it PartialTrimmedDense}  algorithm allows at most $k$ non-skilled nodes in the component whereas {\it CompleteTrimmedDense} attempts to remove as many non-skilled nodes as possible. The smallest resulting component with the required skilled nodes is then picked. This helps reduce the size of the solution, which is now a single component, and hopefully still sufficiently dense since the heuristic started with a 3-approximation to the density objective. 
\begin{algorithm}[]
\caption{PartialTrimmedDense($G, T$)}
\label{algo:ptDlk}
\begin{algorithmic}[1]
\STATE (Note: $T = \{<a, k>\}$)
\STATE $G' \leftarrow$  {\it s-DensestAlk(G, T)}
\STATE ${\cal C'} \leftarrow {\it EnhanceComponent}(G',T)$
\FOR {each component $C'_i \in {\cal C'}$}
\STATE $Q \leftarrow \{ u \mid u \in C'_i \mbox { and } u \not \in S(a) \}$
\WHILE {$Q$ not empty and $| V(C'_i) - S(a) | > k$}
\STATE $u_{min} \leftarrow$ pop lowest degree node from $Q$
\IF {($C'_i - u_{min}$) is connected}
\STATE $C'_i \leftarrow C'_i - u_{min}$
\ENDIF
\ENDWHILE
\IF {$| V(C'_i) - S(a) | > k$}
\STATE ${\cal C'} \leftarrow {\cal C'} - C'_i$
\ENDIF
\ENDFOR 
\STATE Return $\arg \max_{C'_i \in {\cal C'}} density(C'_i)$ 
\end{algorithmic}
\end{algorithm}
\begin{algorithm}[]
\caption{CompleteTrimmedDense($G,T$)} 
\label{algo:ctDlk}
\begin{algorithmic}[1]
\STATE (Note: T $= \{<a, k>\}$)
\STATE $G' \leftarrow$  {\it s-DensestAlk(G,T)}
\STATE ${\cal C'} \leftarrow {\it EnhanceComponent}(G', T)$
\FOR {each component $C'_i \in {\cal C'}$}
\STATE $Q \leftarrow V(C_i) - S(a)$
\WHILE {$Q$ is not empty}
\STATE $u_{min} \leftarrow$ pop lowest degree node from $Q$
\IF {($C'_i - u_{min}$) is connected}
\STATE $C'_i \leftarrow C'_i - u_{min}$
\ENDIF
\ENDWHILE
\ENDFOR 
\STATE Return $\arg \min_{C'_i \in {\cal C'}} | V(C'_i) |$
\end{algorithmic}
\end{algorithm}

\subsection{Single Skill Team Formation}
We run the single skill experiments for $k \in \{3, 5, 7, 9, 11, 13, 15\}$. For each value of $k$, we have a separate run for each skill $a \in \{\textsc{t, ai, db, dm}\}$. 
We calculate statistics, such as density, size, and number of connected components for each solution and present the mean over these four runs as the final statistic. 

Figures~\ref{fig:kSingle}(a) and~\ref{fig:kSingle}(b) show ($k$ vs. density) and  ($k$ vs. size) plots, respectively. From these plots, we can see that the density obtained by {\it s-DensestAlk} significantly outperforms the density obtained by {\it MinDiameter} algorithm. This is of course expected. However, the downside is that the size of the solution to {\it s-DensestAlk} is also larger (and in some cases disconnected). 
The heuristic {\it EnhancedDense} essentially adds neighbors to each component in the solution so that the resulting component satisfies the required skill-set and then picks the one with the smallest size. Therefore connectivity is guaranteed. Further, the reduction in density is not much and even the cardinality has reduced compared to the original solution. This also means that the solution returned by {\it s-DensestAlk} contained a good component to start with - by good component we mean a component that has most of the skills satisfied and has high density. 

\begin{figure*}[t]
\begin{center}
\subfigure[$k$ vs. density]{\includegraphics[angle=270, scale=0.20]{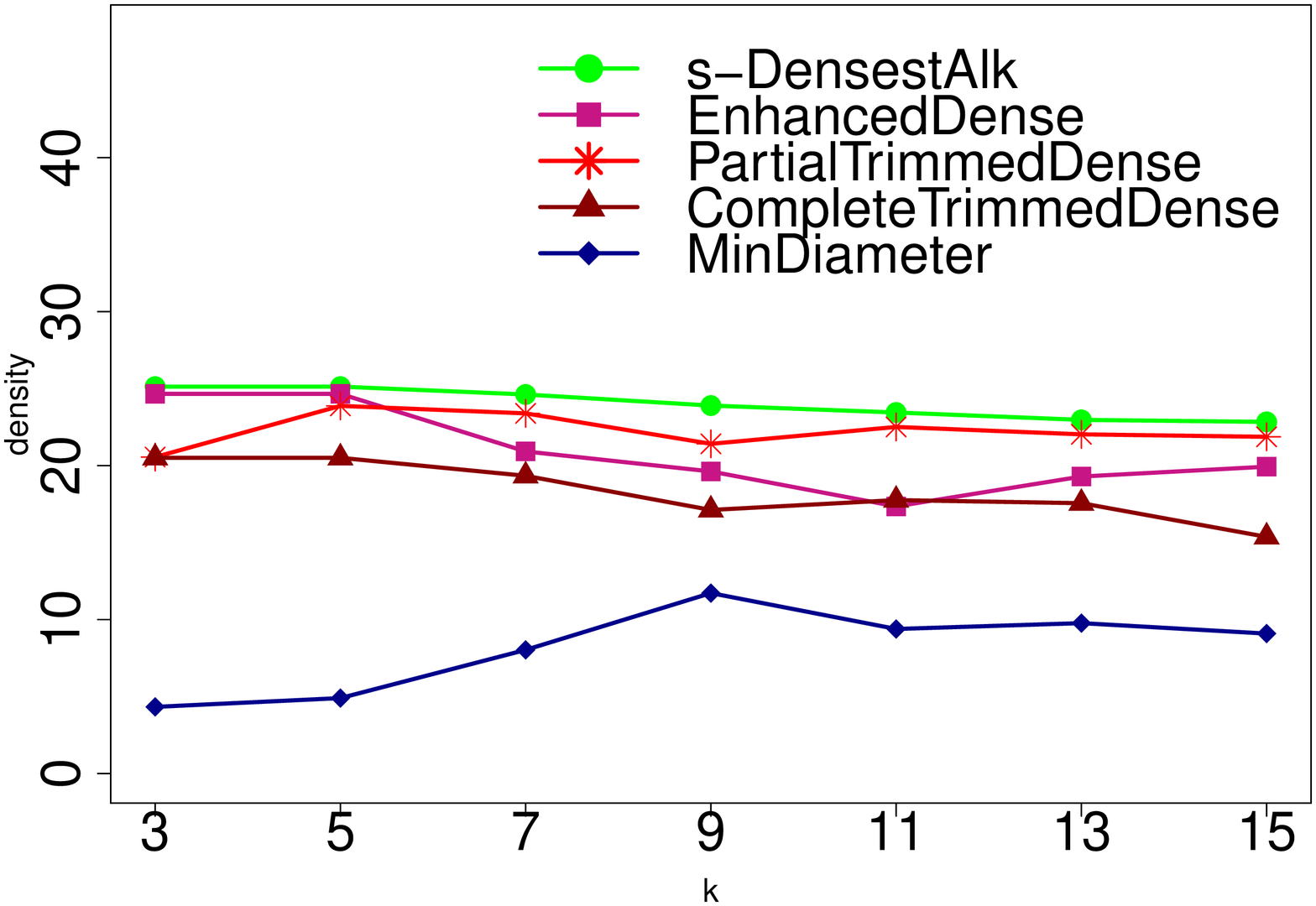}}
\subfigure[$k$ vs. size]{\includegraphics[angle=270, scale=0.20]{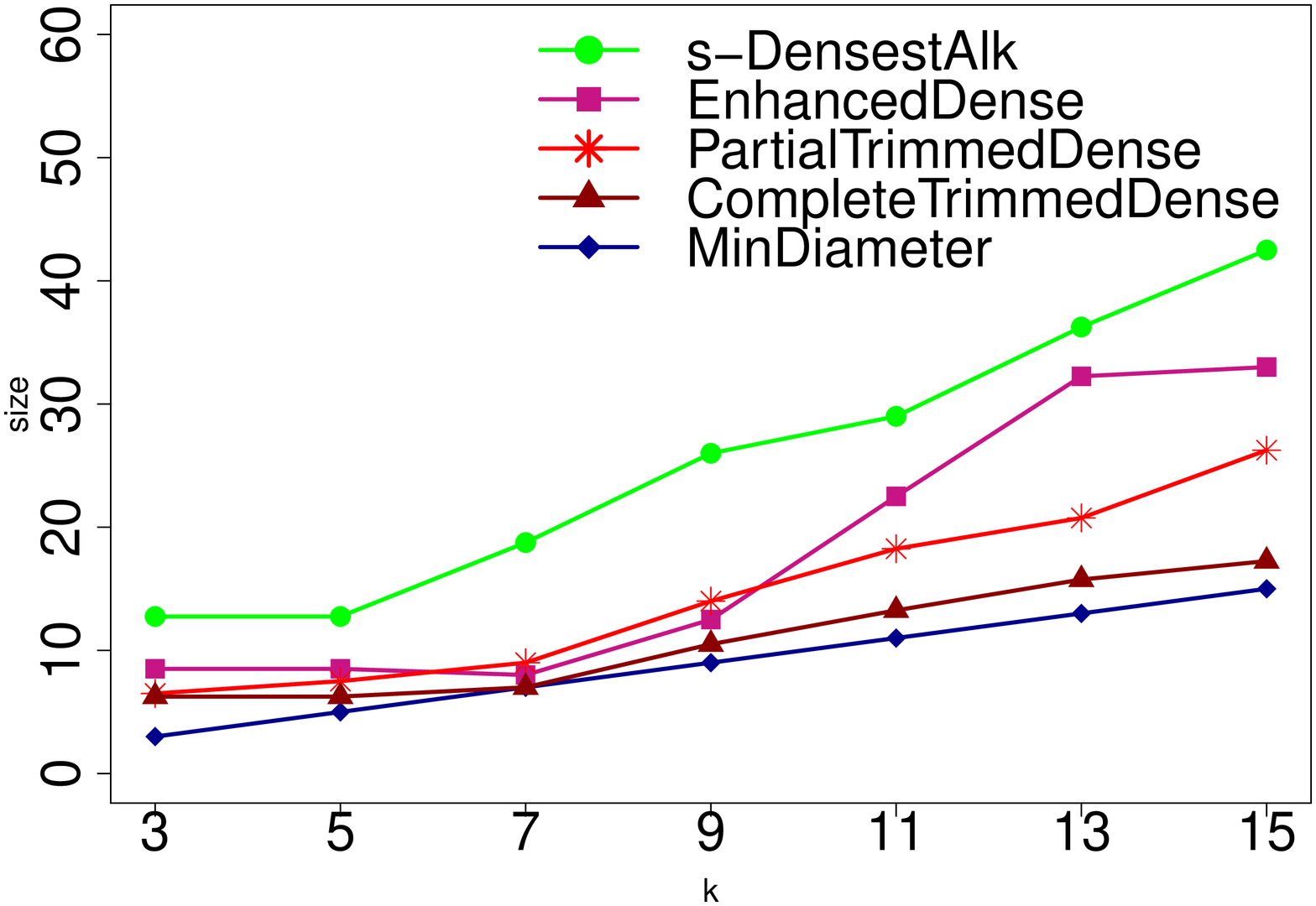}}
\subfigure[$k$ vs. density per node]{\includegraphics[angle=270, scale=0.20]{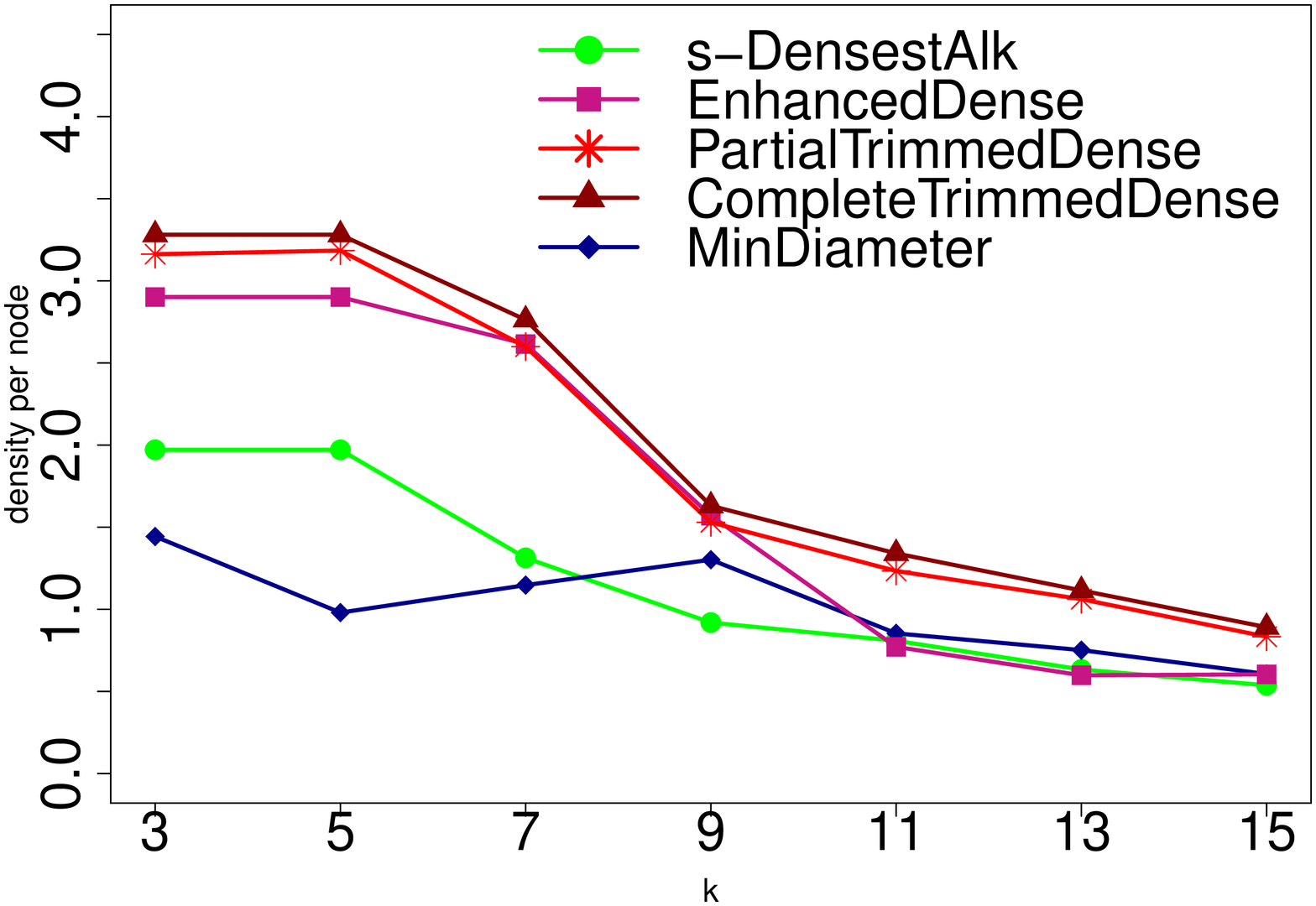}}
\caption{single skill experiments}\label{fig:kSingle}
\end{center}
\end{figure*}

Now, notice that by applying heuristics {\it PartialTrimmedDense} and {\it CompleteTrimmedDense}, we attempt to remove the non-skilled nodes one by one from each of these enhanced components (while maintaining connectivity). As the plots show again, this serves the purpose of significantly reducing the cardinality of the solution and as a hard constraint the algorithm still satisfies the skill requirement. 
It can be observed from the plots that {\it PartialTrimmedDense} has density almost equal to the {\it s-DensestAlk} and the cardinality is reduced by more than fifty percent. Further, {\it CompleteTrimmedDense} gives a solution that has cardinality almost equal to $k$ (which would be optimal), with very little reduction in density. Finally, we plot ($k$ vs. density per node) in Figure~\ref{fig:kSingle}(c). While this figure can be deduced, we present it to highlight the observation that the heuristics reduce the cardinality without compromising on the density. Notice that in this plot, {\it CompleteTrimmedDense} has the highest value of density per node, for every value of $k$. 

Given that density is intuitively a better measure of team collaboration, these results show that we are completely able to eliminate connectivity issues inherent in this objective, and output small yet sufficient, and highly collaborative (dense) teams. 


\subsection{Multiple Skill Team Formation}
We run the multiple skill experiments for $k \in \{3, 8, 13, 18, 23, 28 \}$ and for each run, we randomly choose $k$ skills from $\cal A = \{ \textsc{t, ai, db, dm} \}$. For example, when $k=3$, we may choose a skill (multi)set \textsc{\{t, t, dm\}} which means we want a subgraph that contains at least two authors of skill $T$ and one author of skill $DM$. Recall that a given author can have multiple skills and therefore the solution may consist of a subgraph whose size is less than the value of $k$. 

\begin{figure}[h]
\begin{center}
\subfigure[$k$ vs. density]{\includegraphics[angle=270, scale=0.2]{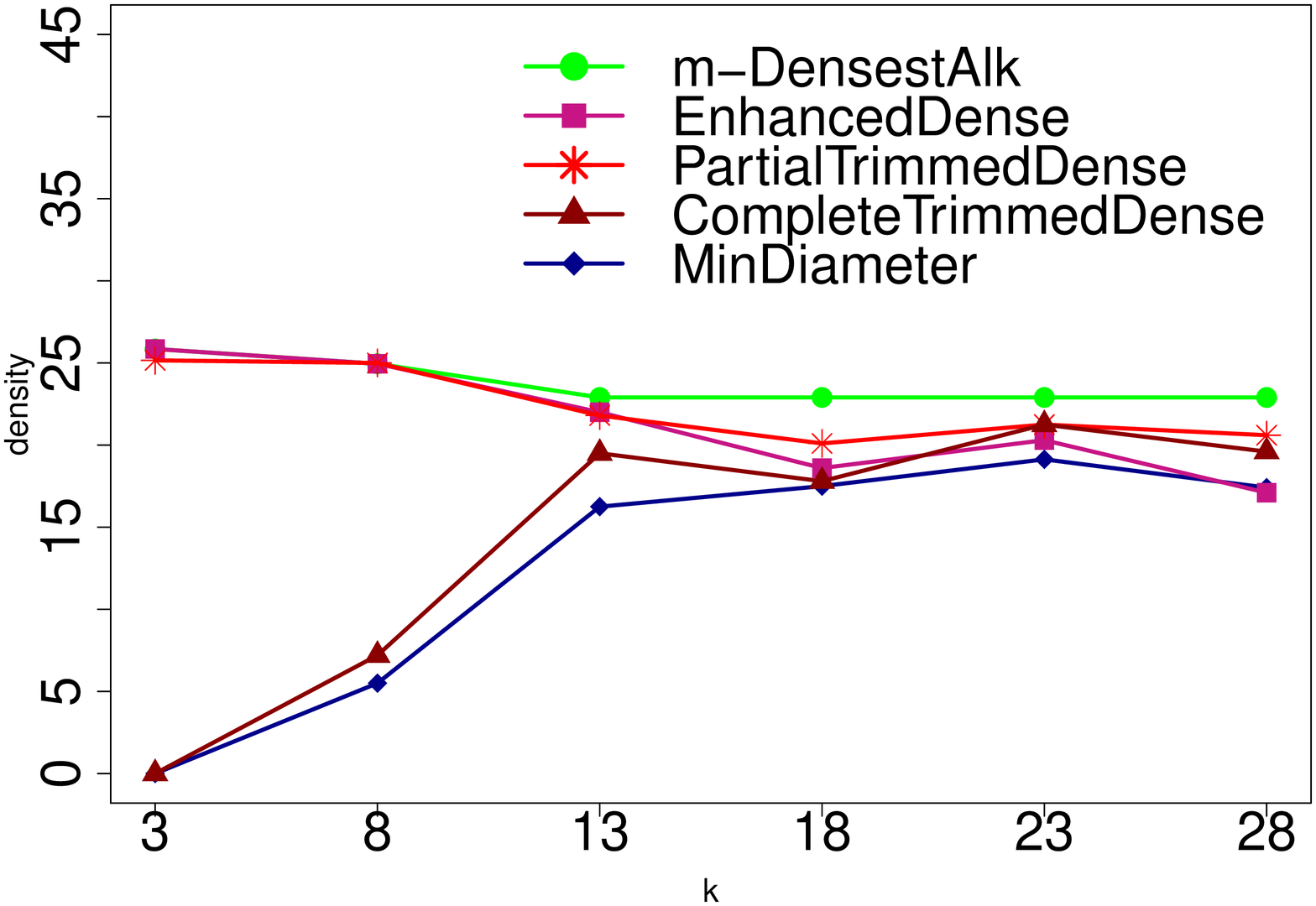}}
\subfigure[$k$ vs. size]{\includegraphics[angle=270, scale=0.2]{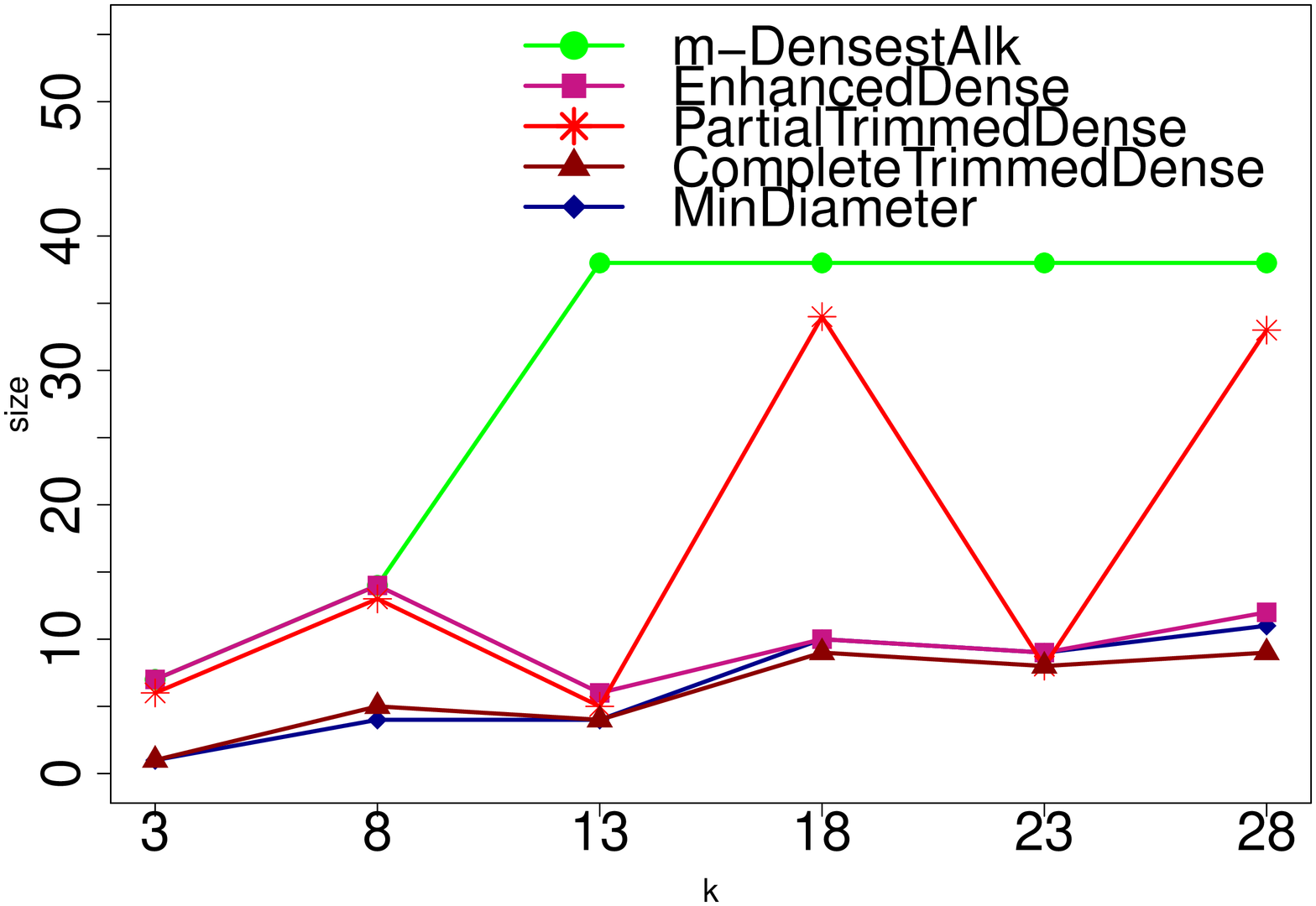}}
\caption{multiple skills experiments}\label{fig:kMulti}
\end{center}
\end{figure}

Figures~\ref{fig:kMulti}(a) and~\ref{fig:kMulti}(b) plots ($k$ vs. density) and  ($k$ vs. size), respectively, for multiple skill team formation experiments. 
Note that the plots for multiple skill experiments fluctuate more than single skill experiments. This is due to the randomness in picking the multiple skills requirements. Also, some solutions returned are of the same size even as $k$ is increased. This is because sometimes the same solution satisfies different required skill sets.

In these figures, we again see that {\it m-DensestAlk} algorithm has the highest density.
Note that the solution with density $0$ and size $1$ corresponds to an individual that has all the required skills. 
Further, similar to single skill experiments, we apply the heuristics mentioned earlier in order to get a connected subgraph without compromising on the density much.  Figure~\ref{fig:kMulti}(b) shows that the heuristics have been effective in reducing cardinality. In fact, the cardinality of the solution obtained by {\it CompleteTrimmedDense} is lesser than $k$ because a single individual can satisfy more than one skills. Further, for the $k\geq 13$ tasks, the density achieved by the heuristics is also close to that of {\it m-DensestAlk}. While sometimes certain heuristics have low density (e.g., $k=3$ or $k=8$), all heuristics offer a nice trade-off between size and density (and return connected solutions by design). For each value of $k$, there exists at least one solution with density close to maximum-density and small cardinality. We omit the density per node plot here due to lack of space, and because it can be deduced from Figures~\ref{fig:kMulti}(a), (b).


\subsection{Density Vs. Diameter Analysis}
In the previous sections, we demonstrated the effectiveness of various heuristic algorithms in order to obtain a solution subgraph that is connected, small and dense. The intuition behind suggesting the density as a metric for team collaborative compatibility is that a denser graph has more edges between nodes, resulting in a greater possibility for collaboration. Small diameter does not necessarily guarantee this property. In this section, we consider three metrics for comparing Density and Diameter based approaches: {\it teamPubs, partialTeamPubs}  and {\it teamPubRatio}. The metric {\it teamPubs} defines the number of publications where all the authors of the publication belong to the solution subgraph. {\it partialTeamPubs} defines the number of publications where at least half of the authors of the publication belong to the solution subgraph. These two metrics give a good indication of the collaboration compatibility of reported teams. In addition, we propose another metric {\it teamPubRatio} which is essential for the comparative study because it is affected by not only the team-members' collaboration compatibility but also on the size of the team. In this case, for each publication, say $p'$, we compute the ratio of $\frac{\mid X' \cap A' \mid}{\mid X' \cup A' \mid}$ where $X'$ is the set of authors in the solution subgraph and $A'$ is the set of authors of the publication $p'$. That is, {\it teamPubRatio} measures the Jaccard similarity between a publication's author set and a team's author set. We then take the average of this quantity over all the publications. 

We now describe the details of the evaluation strategy used to calculate these metrics. For both single skill team formation and multiple skill team formation problems, we consider the teams that were proposed as a solution in the previously described experiments. In particular, we consider the solutions reported by the algorithms {\it CompleteTrimmedDense} and {\it MinDiameter}. We choose only {\it CompleteTrimmed-Dense} algorithm for density because it reports the smallest solutions. The goal is to establish that the small teams obtained by {\it CompleteTrimmedDense} also achieve superior results for the three metrics of collaboration compatibility mentioned above. The results of metric evaluation are shown in the plots ~\ref{fig:kDensityVsDiaSingleSkill} and ~\ref{fig:kDensityVsDiaMultiSkill} for single skill and multi skill experiments, respectively. In each plot, value of $k$ is plotted along the $x$-axis and the value of the the metrics for the corresponding solution subgraphs are along the $y$-axis. In case of single skill experiments, for each $k$, the metric value reported is the average of metric values for the solutions corresponding to each of the skills  $\{$ \textsc {t, ai, db, dm}  $\}$. Further, for the metrics {\it teamPubs} and {\it partialTeamPubs} the $y$-axis defines the resulting number of publications whereas for the metric {\it teamPubRatio}, the $y$-axis defines the {\it scaled} ($100000$ times) metric value. From these plots it can be observed that in both single skill and multi skill team formation problems, the algorithm {\it CompleteTrimmedDense} consistently outperforms the algorithm {\it MinDiameter} for all the three metrics. In case of single skill, for each of the three metrics, and for most values of $k$, the metric value for {\it CompleteTrimmedDense} is about twice that of {\it MinDiameter}. In multi skill, the variation is somewhat larger, but {\it CompleteTrimmedDense} consistently displays superior metric values for all cases. Recall that the size of the solution teams by both these algorithms were very similar (and the metric {\it teamPubRatio} does not necessarily benefit with larger team size); therefore, these experiments suggest that density-based team formation leads to teams with better collaborative compatibility than the diameter-based team formation.
\begin{figure*}[h]
\begin{center}
\subfigure[k vs.Number of Publications]{\includegraphics[angle=270, scale=0.30]{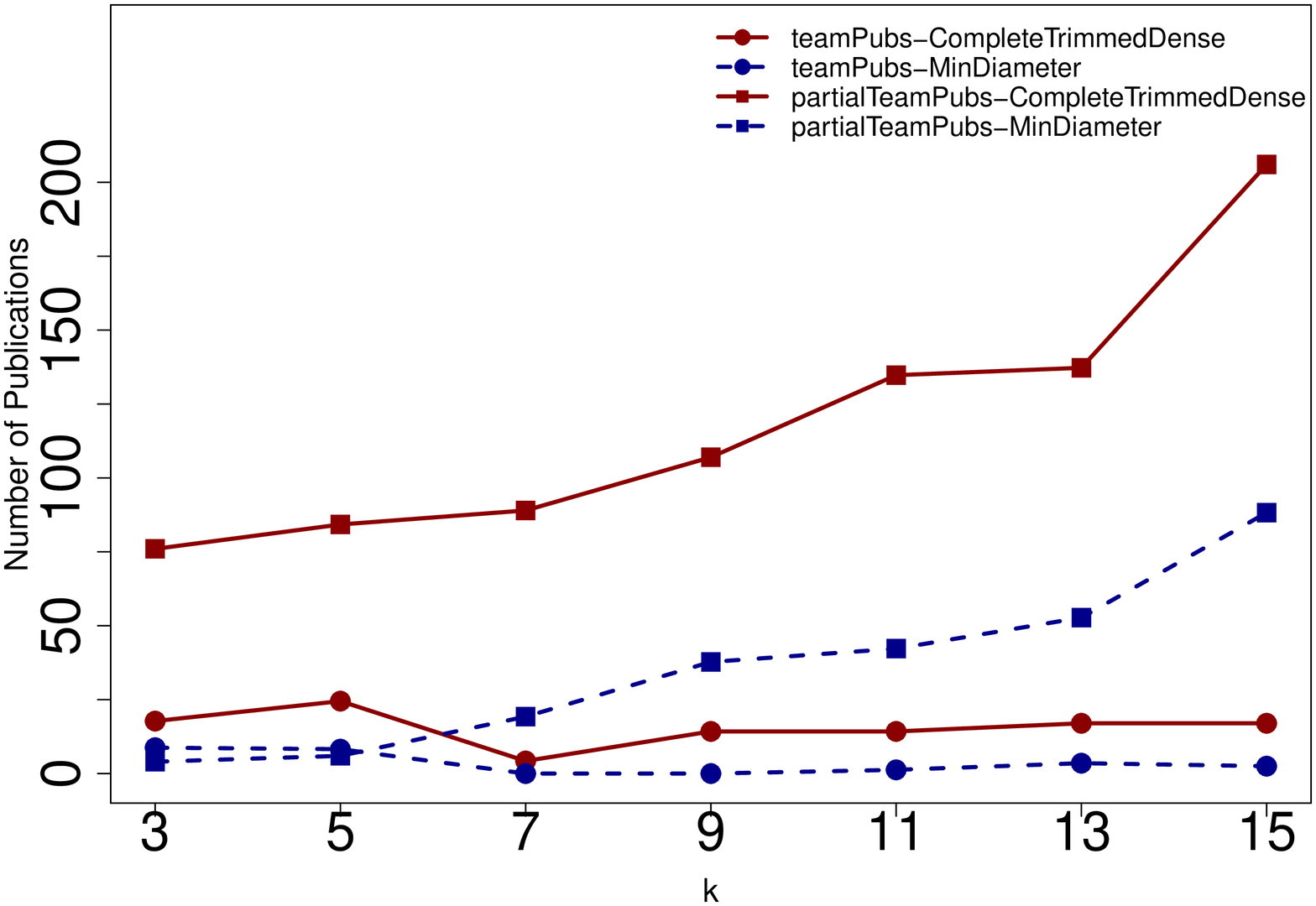}}
\subfigure[k vs. Jaccard Distance]{\includegraphics[angle=270, scale=0.30]{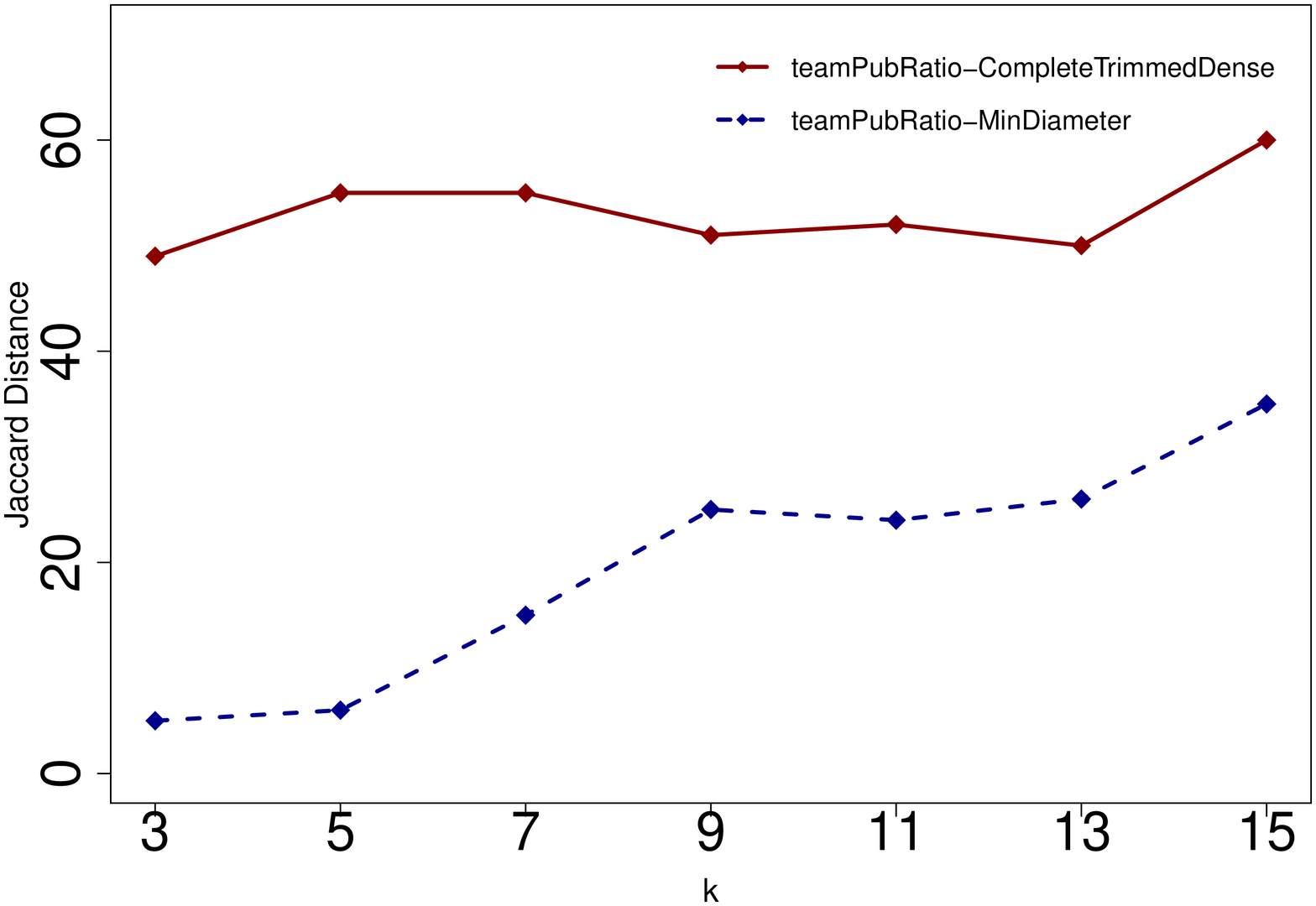}}
\caption{Single Skill Density vs. Diameter Analysis}\label{fig:kDensityVsDiaSingleSkill}
\end{center}
\end{figure*}
\begin{figure*}[h]
\begin{center}
\subfigure[k vs. Number of Publications]{\includegraphics[angle=270, scale=0.30]{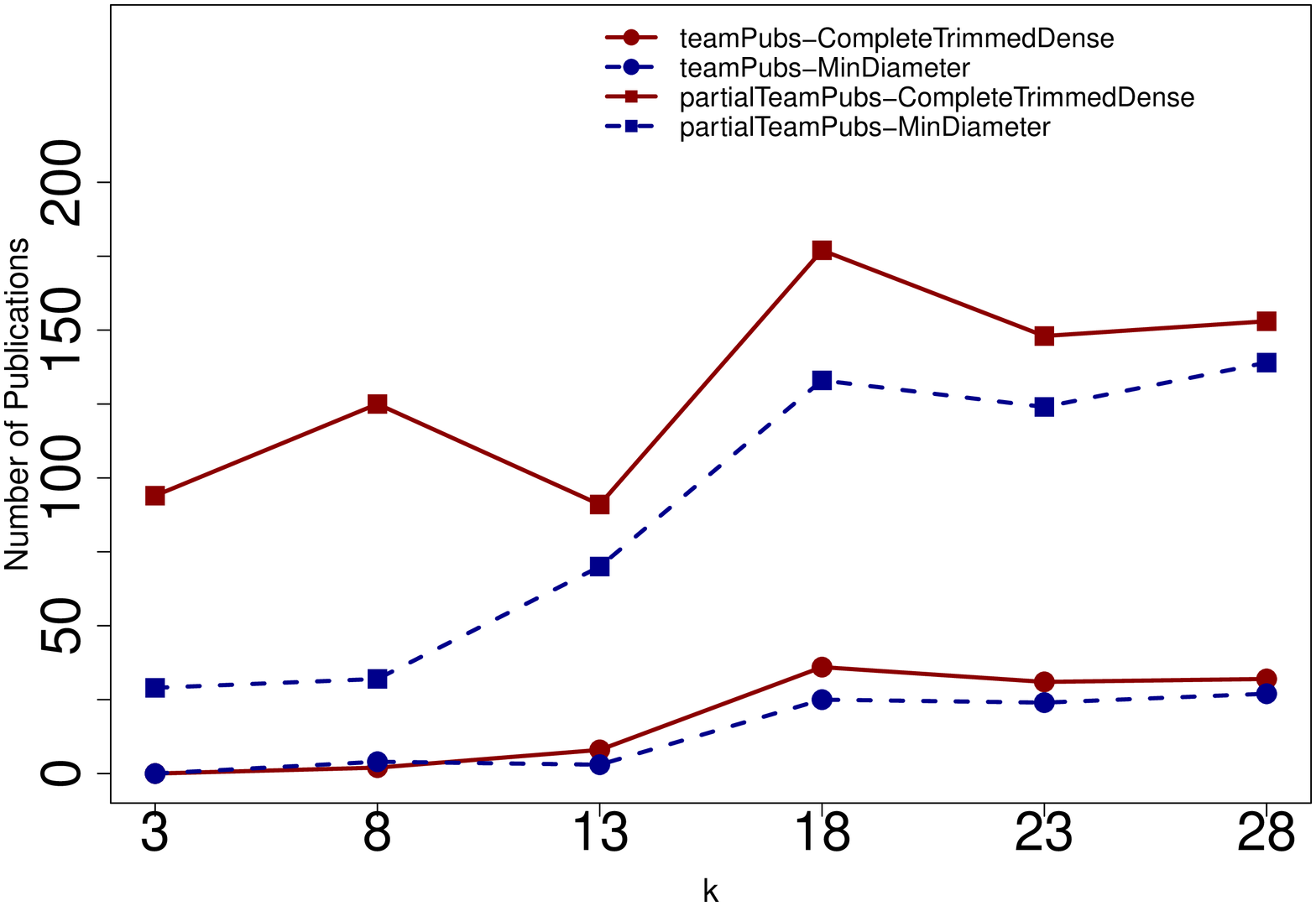}}
\subfigure[k vs. Jaccard Distance]{\includegraphics[angle=270, scale=0.30]{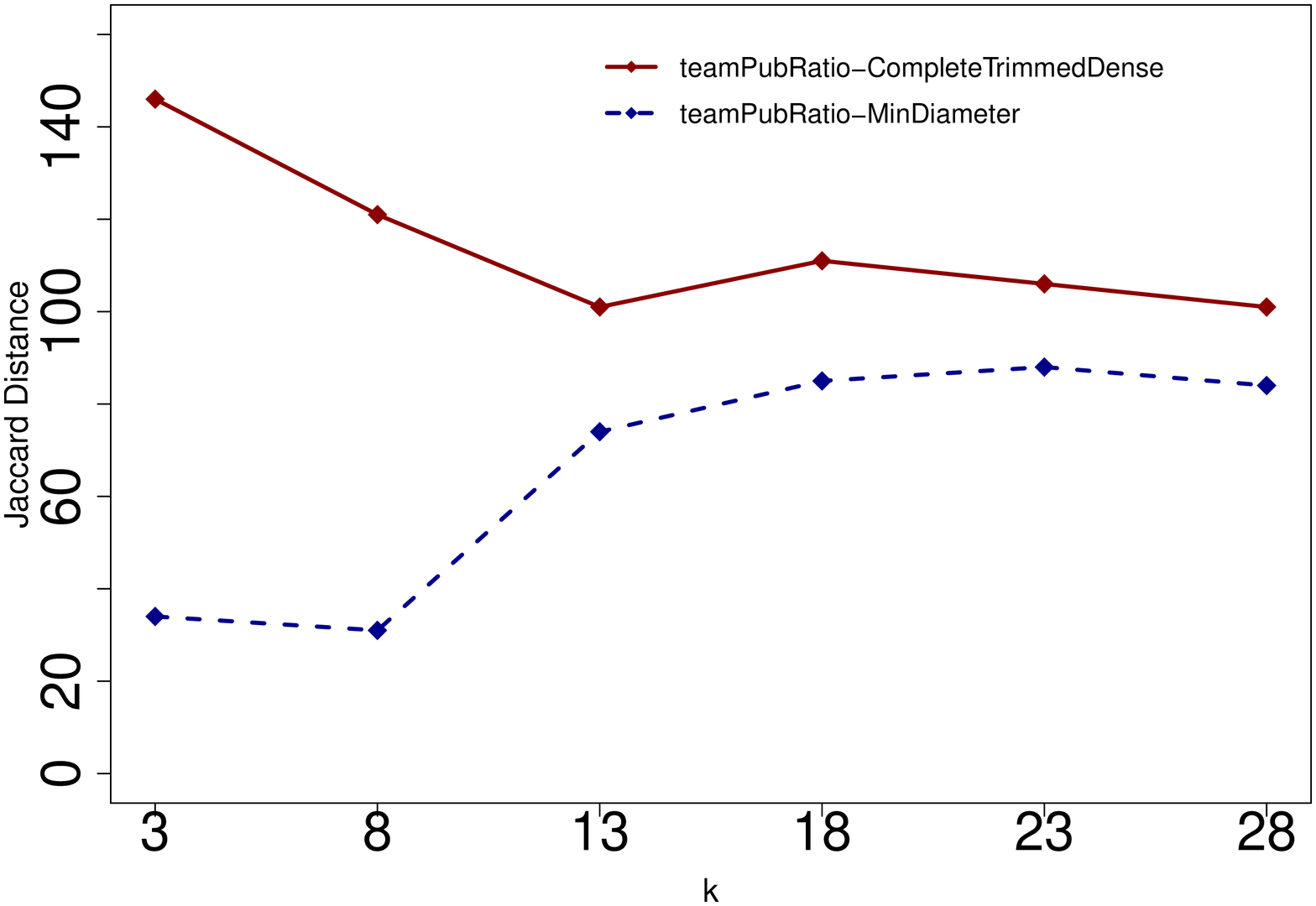}}
\caption{Multiple Skill Density vs. Diameter Analysis}\label{fig:kDensityVsDiaMultiSkill}
\end{center}
\end{figure*}
\subsection{Qualitative evidence}
To analyze the quality of teams that are returned by our algorithms for maximum density, we refer to the {\it Most Cited Computer Science Authors} list maintained by {\it CiteSeerX} (citeseerx.ist.psu.edu/stats/authors?all=true) which contains most cited $10000$ authors. We also refer to the list {\it Central Authors: Computer Science (all-time)} published at (confsearch.org/confsearch/ca.jsp)~\cite{KW}. This list contains $1000$ researchers ranked on the basis of DBLP publications.

We examine the authors of teams returned by {\it s-DensestAlk} and {\it m-DensestAlk} algorithms in order to determine how many authors in the team are among top $500$ and top $1000$ most cited authors according to the list maintained by {\it CiteSeerX}. Due to space constraints, we present only some representative lists from single skill team formation in Table~\ref{QualityAnalysis}. The lists are for $k=3$ for $T$ and $DB$, and for $k=15$ for $DM$ and $AI$.
Team members who appear among the top $500$ and $1000$ cited authors are indicated by bold and italic font, respectively.
We can see from these results that in each team, we have many top cited and prolific/famous authors (who may not be in the top $1000$ list). These results show that teams formed by choosing the objective of maximum density subgraph are {\em intuitively} meaningful. 

\begin{table*} [t]
\caption{Teams reported by s-DensestAlk.}
\label{QualityAnalysis}
\begin{tabular}{lll}
Skills&Authors\\
\hline
T(3)&{\bf Prabhakar Raghavan, Ravi Kumar, Philip S. Yu}, D. Sivakumar, Sridhar Rajagopalan,\\
&Andrew Tomkins \\
DB(3)&{\bf Philip S. Yu,  Haixun Wang, Jiawei Han}, Xifeng Yan, Wei Fan, Hong Cheng,\\
&Charu C. Aggarwal\\
DM(15)&{\bf Jiawei Han, Zheng Chen, Haixun Wang, Philip S. Yu }, Amr El Abbadi,\\
&Benyu Zhang,Wei Fan, Jun Yan, Shuicheng Yan, Hong Cheng, Qiang Yang, Ning Liu, \\
&Jian Pei, Charu C. Aggarwal, Xifeng Yan, Divyakant Agrawal\\
AI(15)&{\bf Ravi Kumar, Ronald Fagin, Philip S. Yu,  Christos Faloutsos,  Zheng Chen},\\
& {\it Wei-Ying Ma, Andrei Z. Broder, Jian-Tao Sun, Hongjun Lu}, Dou Shen,Shuicheng Yan,\\
&Anthony K. H. Tung, Wei Fan, Sridhar Rajagopalan, Qiang Yang, Eli Upfal,\\
&Andrew Tomkins, Jure Leskovec
\end{tabular}
\end{table*}

Complementary results are seen on using the second list, i.e. a list of top $1000$ ranked researchers~\cite{KW}. Instead of presenting another table with author names corresponding to this list, we adopt a different approach for measuring quality. We determine the overall rank of a team using the ranks of the individual authors within the team. To be specific, we compute the mean reciprocal rank of all the skilled individuals in the team and report the final rank of the team as $r = 1000 \frac{\sum_{i}\frac{1}{r_i}}{n_s}$ where $r_i$ denotes the rank of a skilled individual and $n_s$ denotes the skilled individuals in the team. Similar findings are observed if this quantity includes non-skilled nodes as well.
We report the ranks observed in Table~\ref{RankAnalysis}. Our original algorithms for maximum density and the subsequent heuristics form a team of highly ranked authors and perform significantly better than the minimum-diameter algorithm.
The validation of these algorithms over two different qualitative approaches provides further credence to this framework of team formation using a density based objective.

\begin{table}[t]
\caption{Team ranks based on top-ranked authors.}
\label{RankAnalysis}
\begin{tabular}{llll}
Skills&\{s/m\}-&CompleteTrimmed&Min\\
&DensityAlk&Dense&Diameter\\
\hline
T(3)&23.42&8.11&0\\
AI(3)&20.81&17.34&0\\
DB(3)&18.25&18.25&0\\
DM(3)&18.25&18.25&0\\
T(15)&14.95&19.67&2.05\\
AI(15)&15.25&14.48&1.86\\
DB(15)&10.54&10.80&0.75\\
DM(15)&9.55&9.93&1.05\\
T(1),DB(1),&18.25&100&24.39\\
DM(1)&&&\\
T(8),AI(6),&9.49&6.3&4.1\\
DB(8),DM(6)&&&\\
\end{tabular}
\end{table}

\section{Conclusions and Future Work}

We presented a novel approach for skilled collaborative team formation based on finding dense subgraphs. On the theoretical front, we showed constant factor approximation algorithms. On the practical side, we showed several heuristic improvements to our main provable algorithm, and compared it to the previous approach based on identifying small diameter subgraphs. Our experimental results show that the densest subgraph approach significantly outperforms the previous techniques on multiple different measures of collaborative compatibility. 

The formulations in this paper as well as~\cite{LLT} assume that for any given skill, each node in the network is either skilled or not skilled. A nice generalization would be to consider a range of expertise for any skill, modeled as a value between $0$ and $1$.
Another specific open question is to present more efficient 
algorithms for all objectives. 
%
Further, these definitions can be extended along many dimensions. 
In reality a team's 
value depends
on several complex assets such as
cultural backgrounds, geographical location, personalities, ability to work in teams etc. 
Some of these characteristics cannot even be measured easily. 
Yet, while the current models are a good start, it would be nice to investigate these directions and move closer to the motivating realistic scenario.

{
\bibliographystyle{abbrv}
\bibliography{arxiv1}
}

\end{document}